\documentclass[journal]{IEEEtran}
\pdfoutput=1

\usepackage{float}

\usepackage{cite}

\usepackage{tikz}

\usepackage{graphicx}
\ifCLASSINFOpdf
\else
\fi
\usepackage{xcolor,colortbl}
\usepackage{hyperref}
\usepackage{amsmath}
\usepackage{amssymb,amsmath,mathrsfs,dsfont,bbold,latexsym,enumerate,pstricks,epsf}%

\usepackage{algorithm,algcompatible} 
\usepackage{algpseudocode}

\usepackage{array}

\ifCLASSOPTIONcompsoc
  \usepackage[caption=false,font=normalsize,labelfont=sf,textfont=sf]{subfig}
\else
  \usepackage[caption=false,font=footnotesize]{subfig}
\fi

\ifCLASSOPTIONcaptionsoff
  \usepackage[nomarkers]{endfloat}
 \let\MYoriglatexcaption\caption
 \renewcommand{\caption}[2][\relax]{\MYoriglatexcaption[#2]{#2}}
\fi
\usepackage{url}

\usepackage{amsthm}

\theoremstyle{definition}
\newtheorem{definition}{Definition}
\begin{document}
%
\title{Variational Multi-Task MRI Reconstruction: Joint Reconstruction, Registration and Super-Resolution}

%
%
%

\author{{Veronica Corona,~Angelica I. Aviles-Rivero,~No\'emie Debroux,~Carole Le Guyader,~\\Carola-Bibiane Sch\"onlieb}
\thanks{V Corona, N Debroux and CB Sch\"onlieb are with the Department of Applied Mathematics and Theoretical Physics, University of Cambridge, UK. \{vc324,nd448,cbs31\}@cam.ac.uk}
\thanks{
AI Aviles-Rivero is with the Department of Pure Mathematics and Mathematical Statistics, University of Cambridge, UK.  ai323@cam.ac.uk }
\thanks{
C Le Guyader is with the Normandie Universit\'e, INSA de Rouen, France. carole.le-guyader@insa-rouen.fr}
}

\newtheorem{theorem}{Theorem}[section]

\maketitle

\begin{abstract}
Motion degradation is a central problem in Magnetic Resonance Imaging (MRI). This work addresses the problem of how to obtain higher quality, super-resolved motion-free, reconstructions from highly undersampled MRI data. 
In this work, we present for the first time a variational multi-task framework that allows joining three relevant tasks in MRI: reconstruction, registration and super-resolution. Our framework takes a set of multiple undersampled MR acquisitions corrupted by motion into a novel multi-task optimisation model, which is composed of an $L^2$ fidelity term that allows sharing representation between tasks, super-resolution foundations and hyperelastic deformations to model biological tissue behaviors. We demonstrate  that  this  combination yields to significant improvements over sequential models and other bi-task methods. Our results exhibit fine details and compensate for motion producing sharp and highly textured images compared to state of the art methods.
\end{abstract}



%
\IEEEpeerreviewmaketitle

\section{Introduction}
%
%
%
%
\IEEEPARstart{M}{}agnetic Resonance Imaging (MRI) is a widely used and non-invasive modality that creates detailed images of the anatomical structures of the human body, including undergoing physiological events. It allows radiologists to examine MRI for diagnosis, treatment monitoring and abnormality/disease detection~\cite{brown2015mri}. 
However, a central limitation of MRI is the prolonged acquisition period needed to reconstruct an image~\cite{zaitsev2015motion}.  This constraint is reputed to be a major contributor to image quality degradation, and therefore, compromising the expert interpretation.

Image degradation appears as motion artefacts including blurring effects and geometric distortions~\cite{sachs1995diminishing,zaitsev2015motion}. Therefore, the problem of how to reduce the acquisition time whilst producing high quality images, super-resolved and motion-free, is of a great interest in the community, and it is the problem that we address in this paper. 

In particular, in a dynamic MRI setting, acquisitions with low signal-to-noise ratios or small anatomical structures might be severely degraded, affecting the final expert's outcome~\cite{havsteen2017movement}. These small structures can appear smeared or blurred, and discerning whether these are artefacts or lesions is very challenging for the expert, leading to potential false positive or negative findings \cite{Andre2015}. Moreover, movement distortions are most prominent at contrast edges \cite{Birn2004}.

\begin{figure}[t!]
    \centering
    \includegraphics[width=0.5\textwidth]{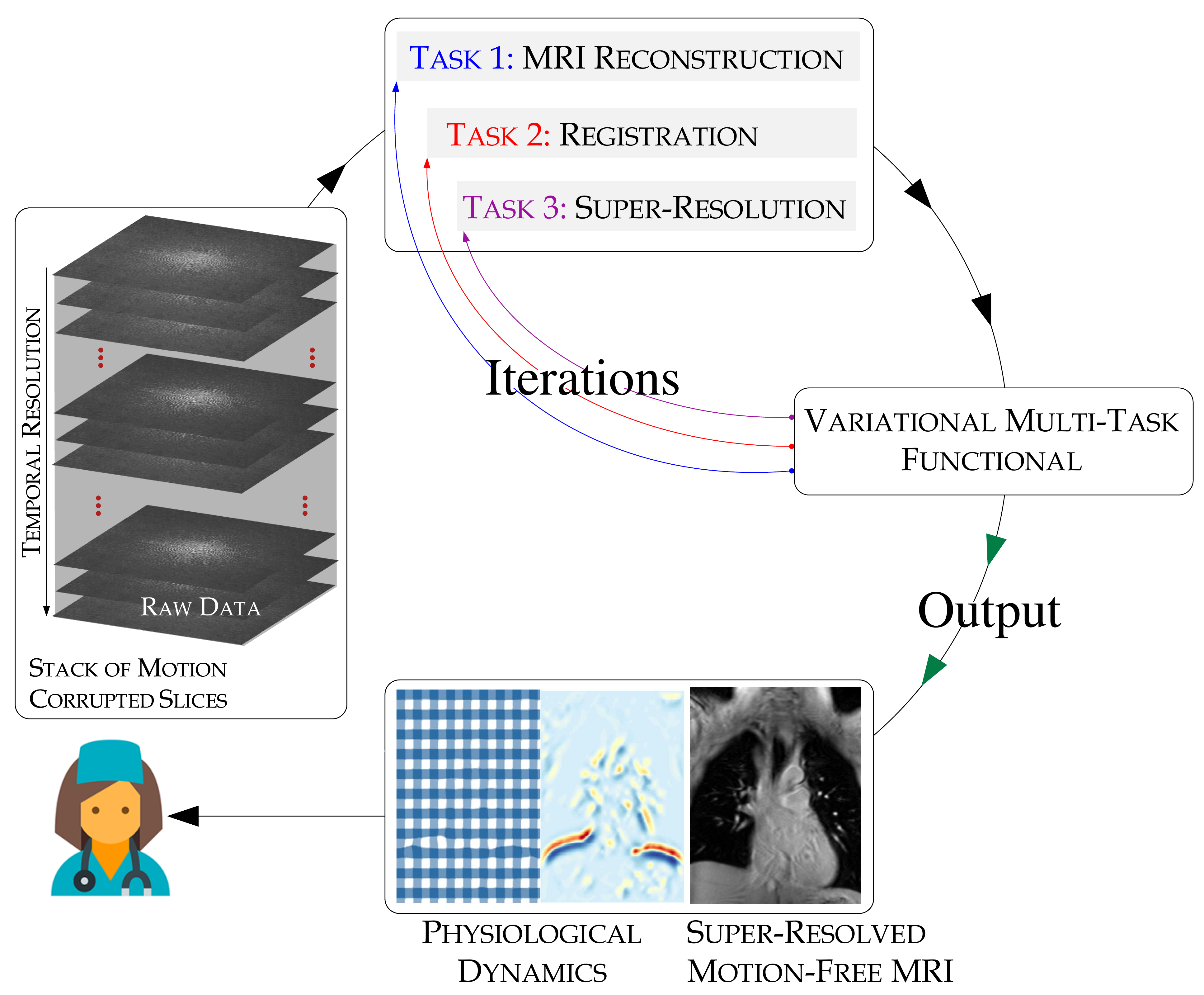}
    \caption{The proposed variational multi-task framework. A set of higly undersampled MRI measurements are taken as input to  our three-tasks framework: reconstruction, registration and super-resolution. We then jointly address them using a proposed functional that has as input a super-resolved motion-free MRI and the physiological dynamics.}
    \label{fig:teaser}
    \vspace{-0.4cm}
\end{figure}

Although, it is possible to reduce the artefacts by performing breath-holding techniques, there is still residual motion to be compensated. This is mainly produced because the timescale of physiological motion is shorter than the required time to form an image. Likewise, gating strategies~\cite{Setser2000,Plein2001,Kaji2001,GEORGE2006924,JIANG2006141,Keall_2000}, which track either the breathing or cardiac cycles, have been also widely explored. However, they are mainly effective for perpetual breathing motion disregarding all other involuntary physiological motion and therefore only partially accurate. Furthermore, it is challenging to precisely co-register these signals to the corresponding MRI data \cite{HOISAK2006339}. 

As an alternative to the aforementioned techniques, a body of research has developed several algorithmic approaches based on the conceptual definition of  Compressed Sensing (CS) which has demonstrated promising results since the seminal paper of Lustig et al.~\cite{lustig2007sparse}. The main idea of using CS is to reconstruct signals from low-dimensional measurements through iterative optimisation relying on sparsity of the image in a transformed domain. Since then, several promising results have been reported in the body of literature e.g.~\cite{liang2007spatiotemporal}, \cite{Lingala::2011}, \cite{lingala2011accelerated}, \cite{Otazo::2015}, \cite{zhang2015accelerating}. However, there is still a need for improving the quality of the MRI reconstruction whilst decreasing the number of measurements. 


A commonality of previous techniques is that they perform  a single task (just reconstruction). However in most recent years, there has been a great interest for improving medical image reconstruction~\cite{Lingala::2015,Royuela-del-Val2016,aviles2018compressed} by using what is called multi-tasking models (also known as joint models). The central idea of this perspective is that by sharing representation between tasks and carefully intertwining them, one can create synergies across challenging problems and reduce error propagation, which results in boosting the accuracy of the outcomes whilst achieving better generalisation capabilities. 

Following the multi-task perspective, different works have been presented e.g.~\cite{Lingala::2015,Royuela-del-Val2016,odille2016joint,aviles2018compressed,blume2010joint,jacobson2003joint}. Unlike existing approaches from the literature, and to the best of our knowledge, we are presenting for the first time a model that considers more than two tasks (see Fig. \ref{fig:teaser}). In this work, we introduce a new variational multi-tasking framework that integrates, in a single model, three relevant tasks in MRI: reconstruction, registration and super-resolution. Whilst this is a relevant part of this work, our contributions are:

\begin{itemize}
    \item We propose a computationally tractable and mathematically well-motivated variational multi-task framework for motion correction in MRI, in which our novelties largely rely on:
        \begin{itemize}
            \item An original optimisation model that is composed of an $L^2$ fidelity term that allows sharing representations between three tasks (reconstruction, super-resolution and registration);  a weighted total variation (TV) ensuring robustness of our method to intensity changes; a TV regulariser of the highly resolved reconstruction; and a hyperelasticity-based regulariser. We demonstrate that this combination yields to significant improvements over sequential models and existing multi-task methods.  
            \item We show that our optimisation problem can be solved efficiently by using auxiliary variables and then splitting it into sub-problems. We show that this  requires lower CPU time than several methods from the body of literature.
        \end{itemize}
    \item  We extensively evaluate our approach using five datasets and different acceleration factors. We also compare our multi-task framework against existing approaches. Our experiments are further validated by interpretations of experts.
\end{itemize}



\section{Related Work}
There have been different attempts to improve motion correction in MRI from undersampled data. Besides motion prevention techniques such as breath-holding, another set of of algorithmic approaches has been devoted to correct for motion using image-based motion tracking, where one needs an explicit estimation of the motion in between scans. The predominant scheme, in this context, is image registration which aims at finding a mapping aligning a moving image to a reference one. Following this perspective, the body of literature can be roughly classified into rigid (translations, rotations) and deformable registration. 

In the first category, several approaches have been proposed including~\cite{gupta2003fast,adluru2006model,wong2008first,johansson2018rigid}. 
However, physiological motion such as cardiac and respiratory, can hardly be characterised by a simple combination of rotations and translations. To mitigate this limitation, motion correction methods based on deformable registration have been proposed such as~\cite{LedesmaCarbayoKellman2007,ledesma2007motion,li2015expiration,jansen2017evaluation}. However, in a closer look on the aforementioned approaches, a commonality between them is that the algorithmic approaches are performed sequentially. That is - the motion estimation task is executed only after the image reconstruction is computed (from now we refer to this perspective as sequential model). A clear drawback of using this perspective is that the motion estimation highly depends on the quality of the reconstruction as well as on the selection of the reference image.

More recently, a body of research has solved jointly multiple tasks (the so-called multi-task approach) such as image reconstruction and registration in a unified framework.
In particular, in the medical domain and following a variational perspective~\cite{burger2018variational}, different works have been reported using multi-task approaches. These include SPECT imaging~\cite{mair2006estimation,schumacher2009combined}, PET~\cite{blume2010joint} and  MRI~\cite{Lingala::2015,Royuela-del-Val2016,aviles2018compressed}- to name a few.  The works with a closer aim to ours are discussed next. 


Authors in~\cite{jacobson2003joint} and \cite{blume2010joint} proposed a joint model composed of a motion-aware likelihood function and a smoothing term for a simultaneous image reconstruction and motion estimation for PET data. Schumacher et al.~\cite{schumacher2009combined} presented an algorithmic approach that combines reconstruction and motion correction for SPECT imaging. The authors proposed a variational approach that includes a regulariser penalising an offset of motion parameter - to favour a mean location of the target object. However, the major limitation is that they only consider rigid motions. In the same spirit, authors of ~\cite{fessler2009,fessler2010optimization}
proposed a generic joint reconstruction/registration framework. That model is based on a penalised-likelihood functional, which uses a weighted least square fidelity term along with a spatial and a motion regulariser.

Odille et al.~\cite{odille2016joint} proposed a joint model for MRI image reconstruction and motion estimation. This approach allows for an estimate of both intra and inter-image motion, meaning that, not only the misalignment problem is addressed but also it allows correcting for blurring/ghosting artefacts. More recently in the context of deep-learning (DL), a number of methods has been investigated for image registration - e.g.~\cite{yang2017quicksilver,de2017end}. Although, certainly, those approaches deserve attention, their review goes beyond the scope of this paper.

\section{Proposed method}
In this section, we introduce our joint variational framework which addresses simultaneously the following three tasks: MRI reconstruction, registration and super-resolution. We introduce the mathematical formulation as separated tasks and then we show how our novel optimisation model judiciously intertwines them. Finally, we describe the numerical realisation of our approach. 

\smallskip
\textbf{Problem statement.} We remark to the reader the focus of this work. Given a set of multiple undersampled MR acquisitions $\{x_t\}_{t=1}^T$ of low resolution and corrupted by motion, we seek to recover a single high resolved, static and motion-corrected image that represents the true underlying anatomy along with the estimation of the breathing dynamics through deformation maps.


\subsection{Task 1: CS MRI reconstruction}
In particular, in standard dynamic MRI, the acquired data is in a time-spatial-frequency space, i.e. \textit{k,t}-space, which is composed of  $x=(x_{m,t})_{m=1,t=1}^{M,T}\in \mathbb{C}^{M\times T}$ measurements. Therefore, the task of MRI reconstruction from those samples, reads:
\begin{equation}
    x=\mathcal{F}u+\eta,
    \label{eq:forward}
\end{equation}
where $\mathcal{F}: \mathbb{R}^{N\times T}\to \mathbb{C}^{M\times T}$ is the \textcolor{black}{undersampled} MRI forward operator. More precisely, $\mathcal{F}=\mathcal{S}\mathcal{A}$ where $\mathcal{S}$ is a subsampling operator, $\mathcal{A}$ the Fourier operator, and $\mathcal{F}^*:\mathbb{C}^{M\times T}\to \mathbb{R}^{N\times T}$ its adjoint. Moreover, $u\in \mathbb{R}^{N\times T}$ is the stack of reconstructed images, $\eta$ an additive Gaussian noise inherent to the acquisition, and $t$ the temporal coordinate.

The MRI reconstruction task is thus highly ill-posed due to the noise and incomplete measurements. However, \eqref{eq:forward} can be solved by adding prior information and then casting the problem as a CS-based optimisation problem: 
\begin{equation}
    u^*\in \arg \min_u \| \mathcal{F}u-x\|_2^2 + \delta \|\Phi(u)\|_1,
    \label{eq:CSModel}
\end{equation}

\noindent
where the first term, i.e. data fidelity term, ensures consistency with the observed data $x$ whilst $\|.\|_1$ enforces sparsity in the transformed domain given by $\Phi$, and $\delta$ is a parameter balancing the influence of each term.

In this work, we focus on the Total Variation (TV)~\cite{ROF} regulariser, which, imposing edge sparsity, leads to piecewise constant reconstructions. It has shown great potential since early developments in  MRI reconstruction~\cite{lustig2007sparse}. However one can easily replace this regulariser by any other one in a plug-and-play fashion. 

Although,  a large body of literature has shown potential results in the context of undersampled MRI reconstruction using CS or its extended philosophies including~\cite{Lingala::2011,Majumdar::2012,Otazo::2015}, there is a still room for improvement, and in particular for the problem of reconstructing a single high quality image that reflects the true underlying anatomy. This motivates the use of two more tasks $-$ image registration and super-resolution, which are described next.

\subsection{Task 2\&3: When Image Registration Meets Image Super-Resolution}
In a dynamic MRI setting, there are two tasks that show a natural strong correlation: motion estimation and super-resolution. Therefore, our hypothesis is that by unifying these two tasks, one can create synergies leading to error propagation reduction, and therefore, an increase of the image quality.


In a multi-frame variational framework, super-resolution is the problem of restoring a high-resolution image from several low quality images that are corrupted by motion. From a variational perspective, it can be expressed as:
\begin{equation}
    \min_u \sum_{i=1}^m\|\operatorname{DBW_i}u-f_i\|_2^2+\lambda \operatorname{Reg}(u),
    \label{eq:SRRmodel}
\end{equation}

\noindent
where $\operatorname{D}$ and $\operatorname{B}$ are the downsampling and blurring operators correspondingly. Moreover, $\operatorname{W_i}$ models the geometric warp existing between the observed images $f_i$ and the restored image $u$ to correct for motion. Finally, $\operatorname{Reg}(u)$ is a generic regulariser.
In this work, the dowsampling operator is modelled as an averaging window, the blurring kernel is assumed to be Gaussian, and the warping operator is viewed as the deformations from a registration task. Whilst for the regulariser we adopt the TV option, our approach is well-suited for the plug-and-play setting. That is- one can easily replace the TV regulariser with other options.

In particular, for our registration method we have the following. 
Let $\Omega$ be the image domain, i.e. a connected bounded open subset of $\mathbb{R}^2$, and $u:\Omega \rightarrow \mathbb{R}$ be the sought single reconstructed image depicting the true underlying anatomy. We introduce the unknown deformations, between the $t$-th acquisition and the image $u$, as $\phi_t : \bar{\Omega} \rightarrow \mathbb{R}^2$. We remark that the deformations are smooth mappings with topology preserving and injectivity properties. Moreover, let $v_t$ be the associated displacements such that $\phi_t=\mbox{Id}+v_t$, where $\mbox{Id}$ is the identity function. At the practical level, these deformations should be with values in $\bar{\Omega}$, and Ball's results~\cite{ball_1981} guarantee this property theoretically for our model. We also consider $\nabla \phi_t : \Omega \rightarrow M_2(\mathbb{R})$ to be the gradient of the deformation, where $M_2(\mathbb{R})$ is the set of real square matrices of order two.

As MRI images biological soft tissues well-modelled by hyperelastic materials, which allows for large  and smooth deformations while keeping an elastic behavior, we propose to view the shapes to be matched in the registration process as isotropic, homogeneous and hyperelastic materials of Ogden type. This is reflected in our formulation as a regularisation on the deformations $\phi_t$ based on the stored energy function of such a material.

In two dimensions, the stored energy function of an Ogden material, in its general form, is given by the following expression: $W_O(F)=\underset{i=1}{\overset{M}{\sum}}a_i\|F\|_F^{\gamma_i}+\Gamma(\mbox{det}F)$, with $a_i>0$, $\gamma_i\geq 1$ for all $i=1,\cdots,M$ and $\Gamma\, :\, ]0;\infty[ \rightarrow \mathbb{R}$ a convex function satisfying $\underset{\delta \rightarrow 0^+}{\lim} \Gamma(\delta)=\underset{\delta \rightarrow +\infty}{\lim}\Gamma(\delta)=+\infty$, $\|.\|_F$ designating the Frobenius matrix norm. 

\textcolor{black}{F}ollowing \cite{Corona}, we consider the \textcolor{black}{particular} energy: 
\begin{equation}W_{Op}(F)=\left\{ \begin{array}{c}a_1\|F\|_F^4 + a_2\left( \mbox{det}F - \frac{1}{\mbox{det}F}\right)^4 \text{ if } \mbox{det}F>0,\\+\infty\text{ otherwise},\end{array} \right.
    \label{eq:Wop}
\end{equation}

\noindent
with $a_1>0$, and $a_2>0$. Both changes in length and area are penalised and topology preservation is ensured with this formulation. 

\subsection{Variational Multi-Task Model: Reconstruction, Registration and Super-Resolution}
In the body of literature, there have been different attempts of using reconstruction, registration and super-resolution. However, they tackled the tasks separately or jointly but up to two tasks. In this part, we describe, for the first time, how these three tasks can be jointly computed to benefit the final reconstruction. The main idea is to exploit temporal redundancy in the data to compensate for motion artefacts due to breathing and/or involuntary movements whilst increasing the resolution to retrieve finer details in the reconstruction. In particular, we now turn to describe how \eqref{eq:CSModel} and \eqref{eq:SRRmodel} can be solved in a multi-task framework.

Our variational multi-task framework takes three key factors: firstly the hyperelastic regulariser \eqref{eq:Wop}, secondly  a discrepancy measure that joins the reconstruction, super-resolution and the registration tasks, and the TV-based regularisers for reconstruction and super-resolution. Moreover, our model  accounts for intensity changes, this, by modifying the CS-classical TV regulariser for the weighted TV to enforce edge alignment (see Definition in Section V of the Supplementary Material).
\textcolor{black}{From now $\mathcal{F}$ is acting on one single frame.}
We thus introduce weights $g_t$ as the Canny edge detector applied to $G_\sigma * \mathcal{F}^*x_t$ - for each $t=1,\cdots,T$ - where $G_\sigma$ is a Gaussian filter of variance $\sigma$. 

We thus consider the following fidelity term and regulariser for our high-resolved image:
\noindent
\begin{equation}
\begin{aligned}
&E(u,(\phi_t)_{t=1,\cdots,T})= \frac{1}{T} \underset{t=1}{\overset{T}{\sum}}\delta TV_{g_t}(\textcolor{black}{(\mathcal{C}u)} \circ \phi_t^{-1})\\
& +\alpha \operatorname{TV}(u)+\frac{1}{2}\|\mathcal{F}((\mathcal{C}u\circ \phi_t^{-1}))-x_t\|_2^2, 
\end{aligned}
\label{eq:ftreg}
\end{equation}
where $\mathcal{C}=DB$ comes from the super-resolution formulation. The first term of $F$ seeks to align the edges of the deformed  reconstruction $((\mathcal{C}u)\circ \phi_t^{-1})$) with the ones of the different acquisitions, whilst regularising \textcolor{black}{it}. The second quantity aims to get $\mathcal{F}((\mathcal{C}u)\circ \phi_t^{-1})$ close to the acquisitions $x_t$, and thus $\mathcal{F}^*(x_t)$ close to $\mathcal{C}u\circ \phi_t^{-1}$ to correct for motion.  

Our variational multi-task framework is then defined as a combination of \eqref{eq:Wop} and \eqref{eq:ftreg}, which leads to the following minimisation problem:

\begin{equation}
\begin{aligned}
 &\inf_{u,(\phi_t)_{t=1,\cdots,T}} G(u,(\phi_t)_{t=1,\cdots,T}) = E(u,(\phi_t)_{t=1,\cdots,T})\\
 &+\frac{1}{T}\underset{t=1}{\overset{T}{\sum}}\int_\Omega W_{Op}(\nabla \phi_t)\,dx, \\
\Leftrightarrow
 &\inf_{u,(\phi_t)_{t=1,\cdots,T}} \frac{1}{T} \underset{t=1}{\overset{T}{\sum}}\frac{1}{2}\|\mathcal{F}( (\mathcal{C}u)\circ \phi_t^{-1})-x_t\|_2^2 + \alpha \operatorname{TV}(u) \\
 &+ \delta \operatorname{TV}_{g_t}( (\mathcal{C}u) \circ \phi_t^{-1})+\int_\Omega W_{Op}(\nabla \phi_t)\,dx, 
 \label{initial_problem}
 \end{aligned}
 \end{equation}

\noindent
We now introduce the next theorem to set the well-posedness of our model.

\begin{theorem}[Existence of minimisers]
 Let $\mathcal{F}=\mathcal{S}\mathcal{A} : L^2(\mathbb{R}^2) \rightarrow L^2(\mathbb{R}^2)$, $\mathcal{C} : L^1(\Omega')\rightarrow L^p(\Omega)$, be linear bounded and continuous for the strong topology operators with $p\in]1,\frac{8}{5}[$, and $\Omega \subset \Omega'$, $\Omega$ and $\Omega'$ connected bounded open subsets of $\mathbb{R}^2$ with boundaries of class $\mathcal{C}^1$ (verified by the chosen operators). With $\delta,\,\alpha,\,a_1,\,a_2>0$, problem (\ref{initial_problem}) admits minimisers $(\bar{u},(\bar{\phi}_t)_{t=1,\cdots,T})$ on $\mathcal{U}=\{u\in BV(\Omega'),\, \phi_t\in \mathcal{W},\, \forall t=1,\cdots,T\,|\, (\mathcal{C}u)\circ\phi_t^{-1}\in BV_{g_t,0}(\Omega),\, \forall t\in\{1,\cdots,T\}\}$, with $\mathcal{W}=\{\psi \in \mbox{Id}+W^{1,4}_0(\Omega,\mathbb{R}^2)\,|\, \mbox{det}\nabla \psi\in L^4(\Omega),\, \frac{1}{\mbox{det}\nabla \psi}\in L^4(\Omega),\, \mbox{det}\nabla \psi >0 \text{ a.e. on } \Omega\}$.
\end{theorem}
\begin{proof}
The proof can be found in Section V of the supplementary material.
\end{proof}

In the next section, we detail how the proposed model \eqref{initial_problem} can be solved in a computational tractable form.

\subsection{Optimisation Scheme}
The numerical realisation of~\eqref{initial_problem} imposes different challenges due to the nonlinearity and nonconvexity in $\nabla \phi_t$ and the composition $(\mathcal{C}u)\circ \phi_t^{-1}$ in the fidelity term. In this work, we overcome these difficulties by introducing three auxiliary variables $z_t,\, h_t,f_t$, this, to mimic $\nabla \phi_t$, $(\mathcal{C}u)\circ \phi_t^{-1}$ and $h_t$. We then relax our problem using quadratic penalty terms. This leads to the following discretised decoupled problem:

\begin{equation}
\begin{aligned}
    \min_{u,\phi_t,z_t,h_t,f_t}\,& \frac{1}{T}\sum_{t=1}^T \underset{x \in \Omega}{\sum} W_{Op}(z_t(x)) + \frac{\gamma_1}{2} \|z_t - \nabla \phi_t\|_2^2 \\
    &+ \frac{\gamma_3}{2} \|\mathcal{F}h_t -x_t \|_2^2 + \alpha \operatorname{TV}(u)\\
    &+ \frac{\gamma_2}{2} \|(h_t - (\mathcal{C}u) \circ \phi_t^{-1}) \sqrt{\operatorname{det} \nabla (\phi_t)^{-1}}\|_2^2 \\
    &+ \frac{1}{2\theta}\|f_t -h_t\|_2^2 +\operatorname{TV}_{g_t}(f_t).
    \end{aligned}
    \label{eq:jointModel}
\end{equation}

We now can solve our minimisation problem by splitting \eqref{eq:jointModel} into five more computational tractable sub-problems.
We now turn to give more details on each sub-problem.

\tikz\draw[red,fill=red] (0,0) circle (.5ex); {\textsc{Sub-problem 1: Optimisation over $z_t$.}}  In practice, $z_t=(z_{t,1},z_{t,2})^T$ simulates the gradient of the displacements $v_t=(v_{t,1},v_{t,2})^T$ associated to the deformations $\phi_t$. For every $z_t$, we have $z_t=\begin{pmatrix} z_{11} ~ z_{12} \\ z_{21} ~ z_{22}\end{pmatrix}$. For the sake of readability, we drop here the dependency on $t$. We solve the Euler-Lagrange equation with an $L^2$ gradient flow and a semi-implicit finite difference scheme and update $z_t$ as:%

\begin{equation*}
\begin{aligned}
z_{11}^{k+1}&=\frac{1}{1+dt \gamma_1} \Bigg( z_{11}^{k} + dt(-4a_1 \|I+z_t^{k}\|^2_F(z_{11}^{k}+1) \\
&- 4 a_2 (1+z_{22}^{k})  c_0 c_1
+ \gamma_1 \frac{ \partial v_{t,1}^{k}}{\partial x} \Bigg), 
   \end{aligned}
\end{equation*}
\begin{equation*}
\begin{aligned}
z_{12}^{k+1}&=\frac{1}{1+dt \gamma_1} \Bigg( z_{12}^{k} + dt(-4a_1 \|I+z_t^{k}\|^2_F z_{12}^{k} \\& + 4  a_2 z_{21}^{k} c_0 c_1
+ \gamma_1 \frac{ \partial v_{t,1}^{k}}{\partial y} \Bigg), 
   \end{aligned}
\end{equation*}
\begin{equation*}
\begin{aligned}
z_{21}^{k+1}&=\frac{1}{1+dt \gamma_1} \Bigg( z_{21}^{k} + dt(-4a_1 \|z_t^{k}+I\|^2_F z_{21}^{k} \\&+ 4  a_2 z_{12}^{k} c_0 c_1
+ \gamma \frac{ \partial v_{t,2}^{k}}{\partial x} \Bigg), 
   \end{aligned}
\end{equation*}
\begin{equation*}
\begin{aligned}
z_{22}^{k+1}&=\frac{1}{1+dt \gamma_1} \Bigg( z_{22}^{k} + dt(-4a_1 \|I+z_t^{k}\|^2_F(z_{22}^{k}+1)\\& - 4 a_2 (1+z_{11}^{k})  c_0 c_1+ \gamma_1 \frac{ \partial v_{t,2}^{k}}{\partial y} \Bigg), 
   \end{aligned}
\end{equation*}

\noindent
with $c_0=\bigg(\operatorname{det}(I+z_t^{k})- \frac{1}{\operatorname{det}(I+z_t^{k})}\bigg)^3$ and $c_1=1+ \frac{1}{(\operatorname{det}(I+z_t^{k}))^2}$.

\medskip
\tikz\draw[red,fill=red] (0,0) circle (.5ex); \textsc{Sub-problem 2: Optimisation over $\phi_t$.} We solve the Euler-Lagrange equation in $\phi_t$, after making the change of variable $y=\phi_t^{-1}(x)$ in the $L^2$ penalty term, for all $t$, using an $L^2$ gradient
flow scheme with a semi-implicit Euler time stepping. 
\begin{equation*}
\begin{aligned}
    0&=-\gamma_1 \Delta \phi_t^{k+1} + \gamma_1 \begin{pmatrix}\operatorname{div}z_{t,1}^{k+1} \\ \operatorname{div}z_{t,2}^{k+1} \end{pmatrix}\\ & + \gamma_2 (h_t^k \circ \phi_t^{k} - \mathcal{C}u^k) \nabla h_t^k(\phi_t^{k}),\\
    \end{aligned}
\end{equation*}

\smallskip
\tikz\draw[red,fill=red] (0,0) circle (.5ex); \textsc{Sub-problem 3: Optimisation over $h_t$.} The update in $h_t$, for all $t$, has a closed form solution using the subsampling operator $\mathcal{S}$ and the Fourier operator $\mathcal{A}$ along with their  adjoints $\mathcal{S}^*$ and $\mathcal{A}^*=\mathcal{A}^{-1}$:
\begin{equation*}
\begin{aligned}
&h_t^{k+1} = \mathcal{A}^{*} \Bigg\{ (\gamma_2\operatorname{det}\nabla(\phi_t^{-1})^{k+1}\mbox{Id}+\gamma_3 \mathcal{S}^*\mathcal{S}+\frac{1}{\theta}\mbox{Id})^{-1}\\
&\bigg(\mathcal{A}\big(\gamma_2 \operatorname{det}\nabla (\phi_t^{-1})^{k+1} (\mathcal{C}u^{k}) \circ (\phi_t^{-1})^{k+1} + \frac{f_t^k}{\theta} \big)  + \gamma_3 \mathcal{S}^* x_t\bigg)
\Bigg\}.
\end{aligned}
\end{equation*}

\smallskip
\tikz\draw[red,fill=red] (0,0) circle (.5ex); \textsc{Sub-problem 4: Optimisation over $f_t$.} This is solved via Chambolle projection algorithm \cite{Chambolle2004}. For an inner loop over $n=1,\cdots,M$:
\begin{equation*}
    \begin{aligned}
    f_t^{n+1}&=h_t^{k+1}-\theta \operatorname{div}p_t^n, \\
    p_t^{n+1}&=\frac{p_t^n+ \delta_t \nabla (\operatorname{div} p_t^n - h_t^{k+1}/\theta)}{1+\frac{\delta_t}{g_t} \| \nabla(\operatorname{div} p_t^n - h_t^{k+1}/\theta)\| },
    \end{aligned}
\end{equation*}
with $\|.\|$ the Euclidean norm. After enough iterations, we set $f_t^{k+1}=f_t^{n+1}$.

\smallskip
\tikz\draw[red,fill=red] (0,0) circle (.5ex); \textsc{Sub-problem 5: Optimisation over $u$.} Finally, using the same change of variables as in sub-problem 2, the problem in $u$ reads:
\begin{equation*}
    \min_u \frac{\gamma_2}{2T} \sum_i^T\|h_t \circ \phi_t- (\mathcal{C}u) \|_2^2+\alpha \operatorname{TV}(u),
\end{equation*}
and we solve it with a primal-dual algorithm: \cite{pd}:
\begin{equation*}
\begin{aligned}
y^{k+1} &= \frac {y^k + \sigma \nabla u^k}{\max (1, \| y^k + \sigma \nabla u^k \| ) }, \\
   u^{k+1}&=\left(\frac{\gamma_2}{T} \mathcal{C}^*\mathcal{C} + \mbox{Id}\right)^{-1}\Big(u^k + \tau \nabla \cdot y^{k+1} \\
   &+ \frac{\gamma_2}{T} \mathcal{C}^* \sum_{t=1}^T h_t^{k+1} \circ \phi_t^{k+1} \Big).
   \end{aligned}
\end{equation*}
\newline

We remark that, in this work, we solve the registration problem in $z_t$ and $\phi_t$ in a multi-scale framework from coarser to finer grids and using a regridding technique \cite{christensen} (see  Supplementary Material Section II). 


\section{Experimental Results}
In this section, we present the experimental results performed to validate our proposed approach.

\subsection{Data Description}
\label{sec:data}
 We evaluate our framework on five publicly available datasets. 
 \begin{itemize}
      \item \textbf{Dataset 1,2 \& 3}\footnote{\url{https://zenodo.org/record/55345\#.XBOkvi2cbUZ}}. 
      These datasets are 2D T1-weighted data~\cite{BAUMGARTNER201783} acquired during free breathing of the entire thorax. It was acquired with a 3T Philips Achieva system with matrix size = $215 \times 173$, slice thickness=8mm, TR=3.1ms and TE=1.9ms. We remark that each dataset refers to three different patients.
     \item \textbf{Dataset 4 \& 5}\footnote{\url{http://www.vision.ee.ethz.ch/~organmot/chapter_download.shtml}}. The datasets are 4DMRI data acquired during free-breathing of the right liver lobe~\cite{Siebenthal2007}. It was acquired on a 1.5T Philips Achieva system, TR=3.1 ms, coils=4, slices=25,  matrix size = $195\times166$, over roughly one hour on 22 to 30 sagittal slices and a temporal resolution of $2.6-2.8$ Hz. 
  \end{itemize}
 
\subsection{Evaluation Protocol} 
To validate our theory, we expensively evaluate our model as follows. 

\medskip
\textbf{Comparison against Sequential Models.} For the first part of our evaluation, we compared our variational multi-task approach against two well-known models,
\textit{rigid} (RIGID), and \textit{hyperelastic} (HYPER), for deformations.
To run this comparison, we  solve the CS reconstruction model with TV, and then register all the frames to a reference frame used as initialisation in our proposed approach. For this, we use the well-established FAIR toolbox~\cite{2009-FAIR}, where we select rigid and hyperelastic transformations. Finally, we perform the super-resolution task with TV.

\begin{figure*}[t!]
    \centering
    \includegraphics[width=1\textwidth]{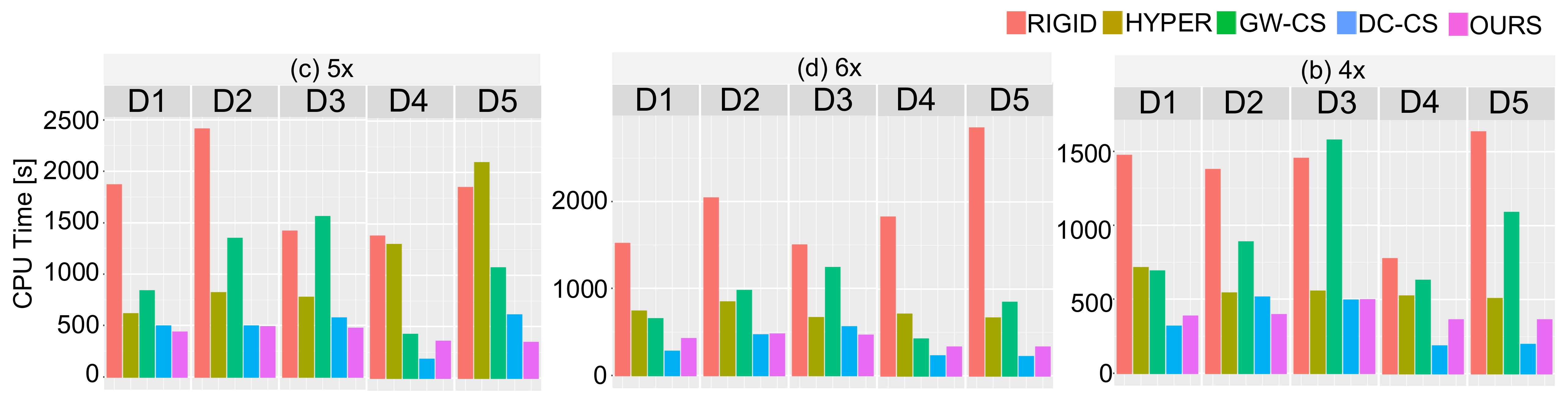}
    \caption{Computational performance comparison between sequential (three tasks), joint (two tasks) and our approach. Elapsed time in seconds. The sequential approaches are definitely much slower than our proposed method. We can see that our approach is comparable and competitive with joint approaches although slightly slower than the DC-CS, which, however, only computes two tasks. }
    \label{fig:cpu}
\end{figure*}

\smallskip
\textbf{Comparison against other Multi-task Approaches.} As to the best of our knowledge, this is the first variational model joining three tasks, we compare our model against two models that only joints two tasks- reconstruction and the motion estimation. More precisely, we compared our method against DC-CS~\cite{Lingala::2015} and  GW-CS~\cite{Royuela-del-Val2016}. To show robustness and generalisation capabilities of our approach, we ran the comparisons using fully sampled data and acceleration factors = $\{2,4,5,6,8\}$.

\begin{figure}[t!]
    \centering
    \includegraphics[width=0.45\textwidth]{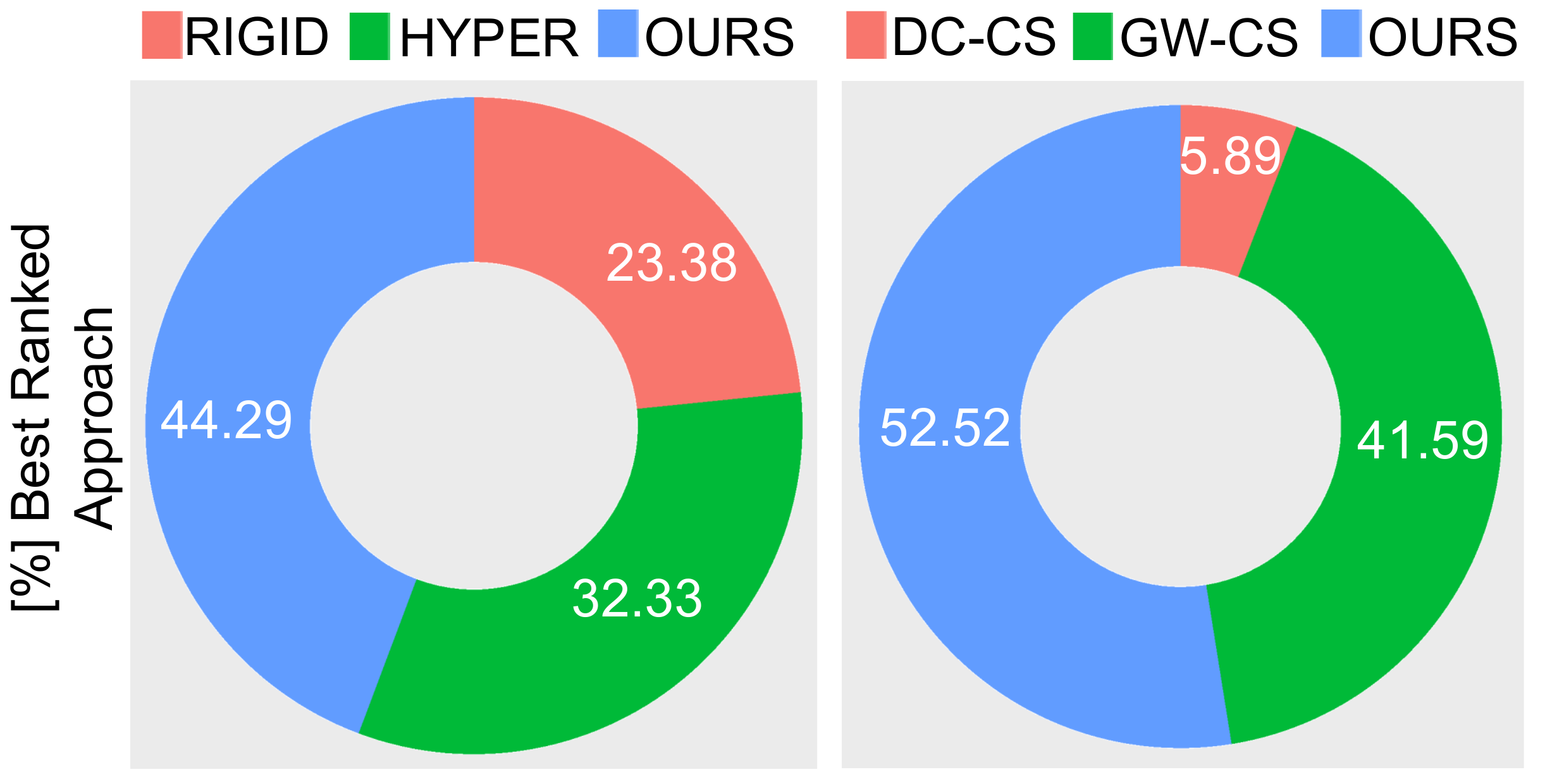}
    \caption{User study results (in \%) indicating the level of agreement of clinicians for sequential and multi-task comparisons. The majority indicates that our proposed reconstructions ranked the best.}
    \label{fig:userstudy}
\end{figure}

\smallskip
\textbf{Metrics Evaluation.} As we seek to recover a single high resolved and motion corrected image, there is not ground truth for this task. Therefore, our evaluation is based on a standard protocol for evaluating MRI reconstruction, that is a user-study (expert scoring). For this, we design a a three-point Likert rating scale in which experts were asked to indicate the level of agreement, ranging from \textit{best reconstruction to worst reconstruction}. The study is also supported by a nonparametric statistical test.
Detailed protocol can be found in Section III of the Supplementary Material. Moreover, to further support our multi-task model, we also offer CPU time comparison against all the compared approaches. 

\smallskip
The experiments reported in this section were run under the same conditions in a CPU-based Matlab implementation. We used an Intel core i7 with 4GHz and a 16GB RAM.

\subsection{Parameter Selection}
In our experiments, we set the parameter of our approach and the compared ones as described next. 
For our experiments, we set the parameters as displayed in Table~\ref{tab::parameter}.  Whilst for   the sequential approaches, based on the FAIR implementation, we set the hyperelastic regularisation parameter for Datasets 1,2 \& 3 $ =1$ and for Datasets 4 \& 5 $ =0.1$.

\subsection{Results and Discussion}
We evaluate our proposed approach following the scheme described in Section IV-B.

\medskip
\noindent
\textbf{$\triangleright$ Is our Multi-tasking Approach Better than a Sequential one?} We start evaluating our approach against two sequential models. We remark to the reader that sequential means to execute  tasks  (reconstruction  registration and super-resolution) one after another. In particular, we compared our approach against two well-known models for deformations:  rigid (RIGID) and hyperelastic (HYPER).  Results of this comparison are displayed in Figs. 4 and 5, and using different acceleration factors.

\begin{table}[t!]
\centering
\begin{tabular}{cccccc}
\cline{2-6}
\textsc{} & \cellcolor[HTML]{EFEFEF}\textsc{D1} & \cellcolor[HTML]{EFEFEF}\textsc{D2} & \cellcolor[HTML]{EFEFEF}\textsc{D3} & \cellcolor[HTML]{EFEFEF}\textsc{D4}  & \cellcolor[HTML]{EFEFEF}\textsc{D5}\\ \hline
\multicolumn{1}{c|}{\cellcolor[HTML]{EFEFEF} $a_1$} & 1 & 1 & 1 &  1 & 1\\ \hline
\multicolumn{1}{c|}{\cellcolor[HTML]{EFEFEF} $a_2$} & 50 & 50 & 50 & 100 & 100 \\ \hline
\multicolumn{1}{c|}{\cellcolor[HTML]{EFEFEF} $\gamma_1$  } & 5 & 5 & 5 & 1 &  1\\ \hline
\multicolumn{1}{c|}{\cellcolor[HTML]{EFEFEF} $\gamma_2$} & $10^5$ & $10^5$ & $10^5$ & $10^5$  &$10^5$ \\ \hline
\multicolumn{1}{c|}{\cellcolor[HTML]{EFEFEF} $\gamma_3$} & $15$ & $15$ & $15$ & $1$  &$1$ \\ \hline
\multicolumn{1}{c|}{\cellcolor[HTML]{EFEFEF} $\theta$} & $5$ & $5$ & $5$ & $5$  &$5$ \\ \hline
\multicolumn{1}{c|}{\cellcolor[HTML]{EFEFEF} $\sigma$} & $1.5$ & $1.5$ & $1.5$ & $2$  &$2$ \\ \hline
\multicolumn{1}{c|}{\cellcolor[HTML]{EFEFEF} $k$} & $2$ & $2$ & $2$ & $2$  &$2$ \\ \hline
\multicolumn{1}{c|}{\cellcolor[HTML]{EFEFEF} $N \& n$} & $500$ & $500$ & $500$ & $500$  &$500$ \\ \hline
\multicolumn{1}{c|}{\cellcolor[HTML]{EFEFEF} $\alpha$} & $0.01$ & $0.01$ & $0.01$ & $0.001$  &$0.001$ \\ \hline
\end{tabular}
\vspace{0.2cm}
\caption{Parameter values using for our model and for all datasets. In this table, "D" stands for Dataset.}
\label{tab::parameter}
\end{table}

In a closer look at those Figs., one can see that our reconstructions have better sharp edges and retrieve fine details, in the heart and below the lung areas,  than the sequential approaches.  Particularly, the rigid transformation is not able to compensate for the whole motion and thus blurring effects are visible, especially under the lungs, for all acceleration factors. Hyperelastic deformations however, have more degrees of freedom and are capable of better compensating for motion which is manifested as sharp edges in the HYPER reconstruction. Moreover, the darker structure, at the center bottom of the heart, disappears or is much less visible in the HYPER reconstructions than in our approach. This effect is observed for all acceleration factors.

\begin{figure}[t!]
    \centering
    \includegraphics[width=0.5\textwidth]{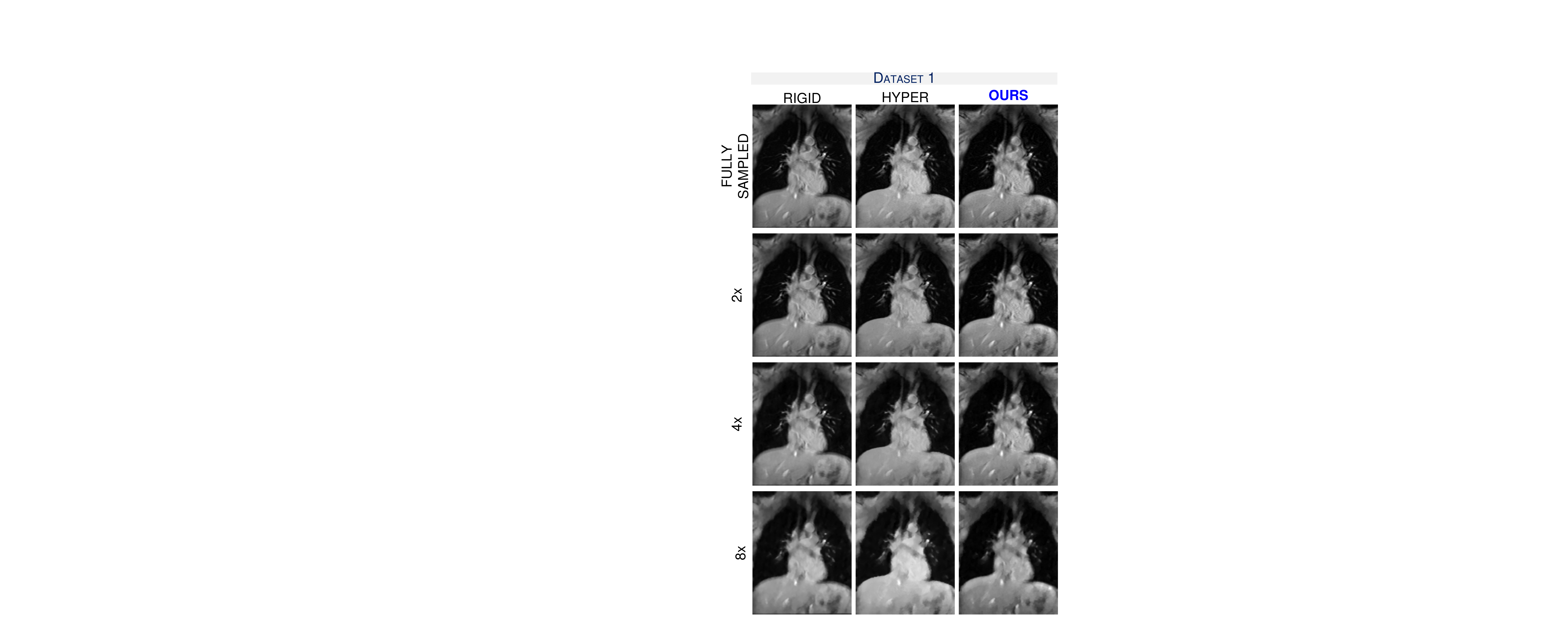}
    \caption{Reconstruction results for Dataset 1 compared to sequential approaches based on rigid and hyperelastic registration, for different acceleration factors. Our proposed reconstruction results in sharp edges and retrieves fine details especially for higher acceleration factors.}
   \label{fig:sD1B}
\end{figure}

\begin{figure}[t!]
    \centering
    \includegraphics[width=0.5\textwidth,  height=14.4cm]{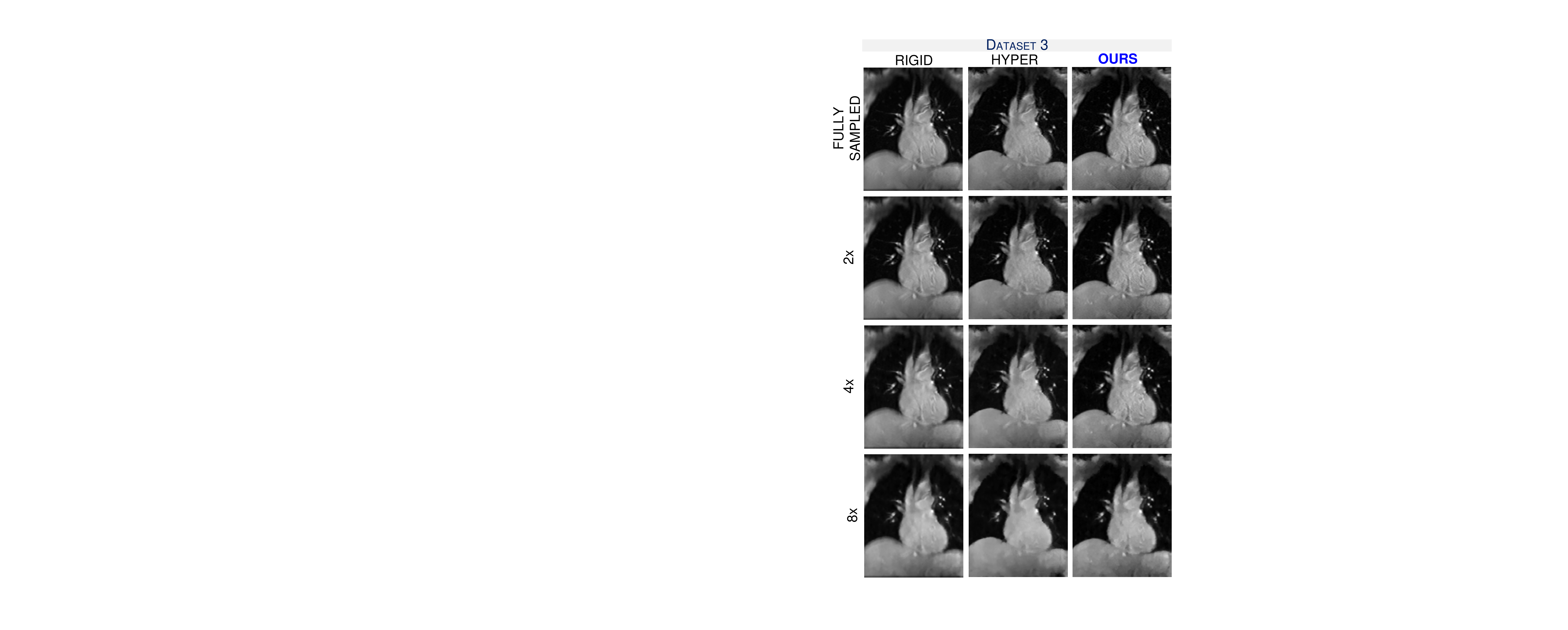}
    \caption{Reconstruction results for Dataset 3 compared to sequential approaches based on rigid and hyperelastic registration, for different acceleration factors. Our proposed reconstruction preserves fine structures and better correct for motion, thus resulting in sharper edges compared to the sequential reconstructions.}
   \label{fig:sD3B}
\end{figure}


Besides, as the acceleration factor increases, the HYPER reconstruction loses the initial contrast, which is particularly visible for the acceleration factor of 8. In contrast our multi-task framework is able to preserve it nicely. This shows the robustness of our method to noise and corrupted data. The benefits of our multi-task framework is prevalent to all datasets (see Section IV of the Supplementary Material for Dataset 2, 4 and 5). 


Figs. \ref{fig:sD1B} and \ref{fig:sD3B} show that hyperelastic deformations are better suited to deal with complex physiological motions, as the RIGID reconstructions exhibits strong blurring artefacts, this, due to residual movements amplified as the acceleration factor increases. Also, our method is able to preserve small structures in the kidney and the white blood vessels in the liver even for large acceleration factors contrary to the sequential HYPER approach. For the acceleration factor of 8, the HYPER reconstruction suffers more from staircaising effects than our approach and loses the initial contrast.

Overall, we can show that sharing representation between tasks (i.e. our multi-task approach) leads to better MRI reconstructions than if one performs the task separately. This is strongly supported by two factors, the computational time and the expert agreement.  Following common protocol for MRI evaluation, we performed a user-study, in which we asked to twelve experts (radiologist trainees and experienced) to evaluate reconstructions with all acceleration factors in Figs. \ref{fig:sD1B} and \ref{fig:sD3B} (see Supplementary material for Datasets 2, 4 and 5). 

The outcome is displayed in Fig.~\ref{fig:userstudy}. At left side of this figure, one can see that overall (i.e. for all reconstruction/all acceleration factors) our approach was ranked best, with a $44.29\%$ of agreement, in comparison with the output from the other methods. We also ran the nonparametric Friedman test, per acceleration and therefore accounting for FDR, and we found that there is significant statistical difference- that is, our approach offered the best reconstructions.  

To further support our model performance, we also analyse the difference maps to assess the quality of our registration, and therefore, its motion correction potential. To do this, we inspect the uncorrected average of the difference image, between a reference frame and each individual one, which is displayed at the left side of  Fig.~\ref{fig:diff_maps}. From this column, we can observe that the motion is significant in both the datasets. However, when we inspect the mean difference between our reconstruction and the individual registered acquisitions ($h_t \circ \phi_t$),  at the middle and right sides of  Fig.~\ref{fig:diff_maps}, one can see that the structures are very well-aligned resulting in a much smaller range in difference maps. Overall, our approach successfully corrects for motion even at low undersampling rates, and this effect is preserved for all datasets. 

\smallskip
\textbf{$\triangleright$ Is it Three-Task Better Than Two-Task Framework?}
In a multi-task framework a key factor is to assess if the tasks are not affecting negatively the final MRI reconstruction. To evaluate this factor, we ran a set of experiments of our approach against two multi-task frameworks DC-CS~\cite{Lingala::2015} and  GW-CS~\cite{Royuela-del-Val2016}. These approaches perform only two-tasks (reconstruction and motion estimation), these approaches are our baselines as, to the best of our knowledge, there exists no approaches that join three tasks. 

\begin{figure}[t!]
    \centering
    \hspace{-0.5cm}
    \includegraphics[width=0.5\textwidth]{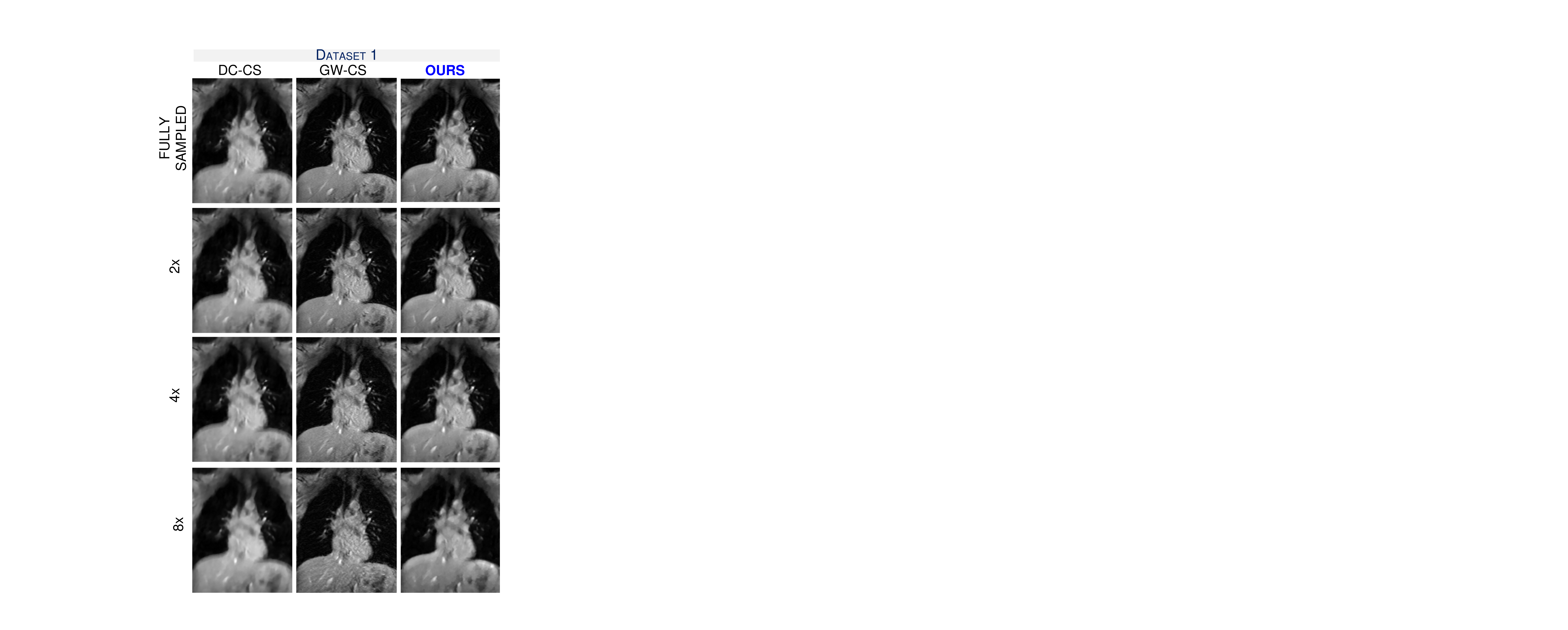}
    \caption{Reconstruction results for Dataset 1 for different acceleration factors and different joint approaches in comparison to our proposed method. We can clearly see that our approach provides the best results in terms of sharp structures and fine texture, while DC-CS results very blurry and GW-CS very noisy. This is particularly accentuated for high undersampling factors.}
    \label{fig:dataArec_joint}
\end{figure}

\begin{figure}[t!]
    \centering
    \hspace{-0.5cm}
    \includegraphics[width=0.5\textwidth, height=14.4cm]{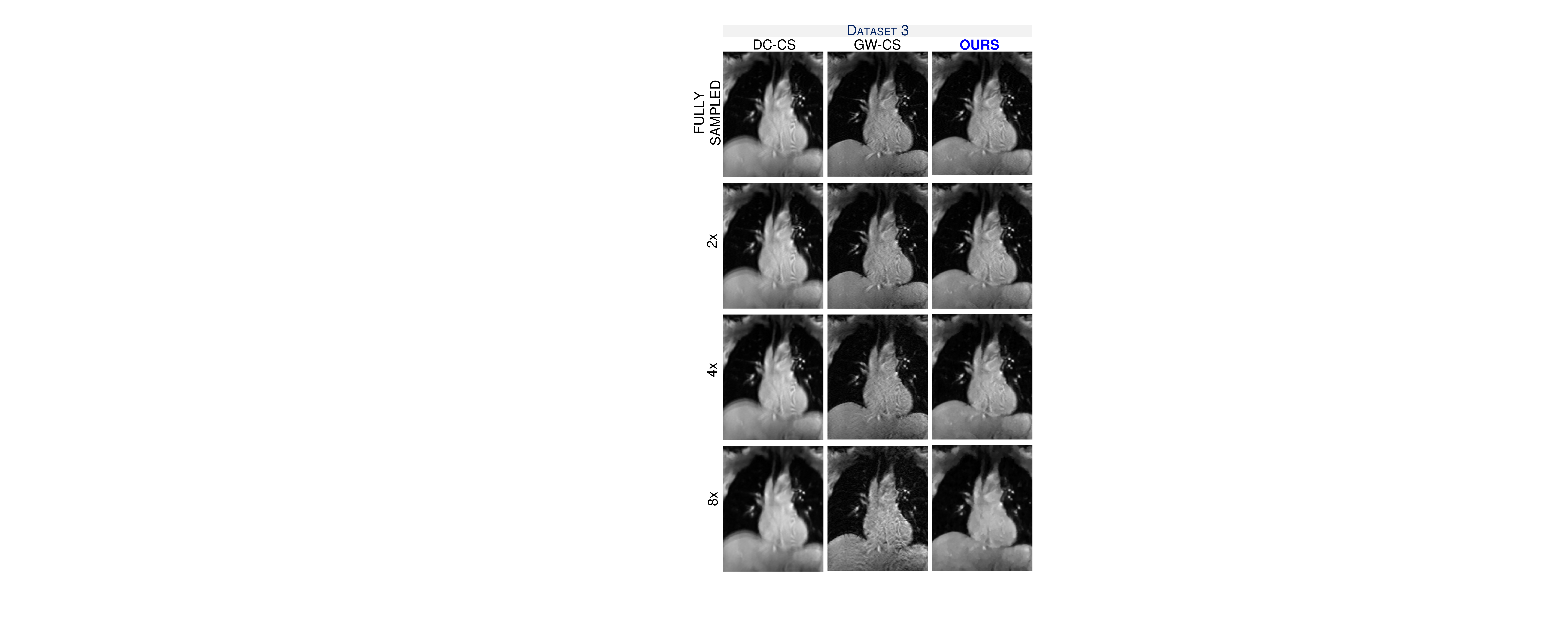}
    \caption{ Comparison of our multi-task framework vs other bi-task approaches on Dataset 5. 
    We note that our approach can preserve fine detail and sharp edges whilst  DC-CS fail to compensate for motion, yielding blurring artefacts. We can also note the GW-CS approach highly amplifies the noise. This is more visible for high undersampling factors. }
    \label{fig:dataDrec_joint}
\end{figure}

The MRI reconstruction from our model against DC-CS and  GW-CS can be seen in Figs. \ref{fig:dataArec_joint} and \ref{fig:dataDrec_joint} (see Supplementary Material Section IV for further visualisation with the remaining datasets). In a closer look at these figures, one can observe very blurred  reconstructions from DC-CS, which can be interpreted as a failure of the model to capture the complex intrinsic nature of physiological motions. In contrast, GW-CS and our reconstructions are sharper even for very low undersampling factors and compensate well for motion.

However, our method is more robust to noise and outliers (as displayed in the compared reconstructions). Although, the GW-CS reconstructions preserve fine textures and small structures, they are noisier than ours. That is, our approach preserve improves in terms of preserving information whilst removing noise in comparison with GW-CS. This effect is elevated even more as the  acceleration factor increases. For example, for an acceleration factor of 8, artefacts and noise are visible in the heart and under the lungs in the GW-CS reconstruction whereas ours is clearer.

Another example of the good performance of our approach can be seen in  Fig.\ref{fig:dataDrec_joint}, in which we are able to retrieve more clinically useful texture and fine details than the  GW-CS technique. 
This is particularly visible in the central part of the heart where noise is visible in the GW-CS reconstructions especially as the acceleration factor increases. 

To further support our results, we display, at the right side of Fig.~\ref{fig:userstudy}, the overall outcome of the user-study. From this plot, we can see that the majority of the expert agreed that our reconstructions are better than the compared approaches. Although, the second best ranked  is GW-CS,  it fails to correct for noise which compromises the readability of the underlying texture.  Moreover, as soon as the acceleration factor increases, the noise level jumps, reducing drastically the readability and interpretability of the GW-CS reconstructions whereas our method retrieves relevant small structures and denoises the reconstruction.

\begin{figure}[t!]
    \centering
    \hspace{-0.5cm}
    \includegraphics[width=0.5\textwidth]{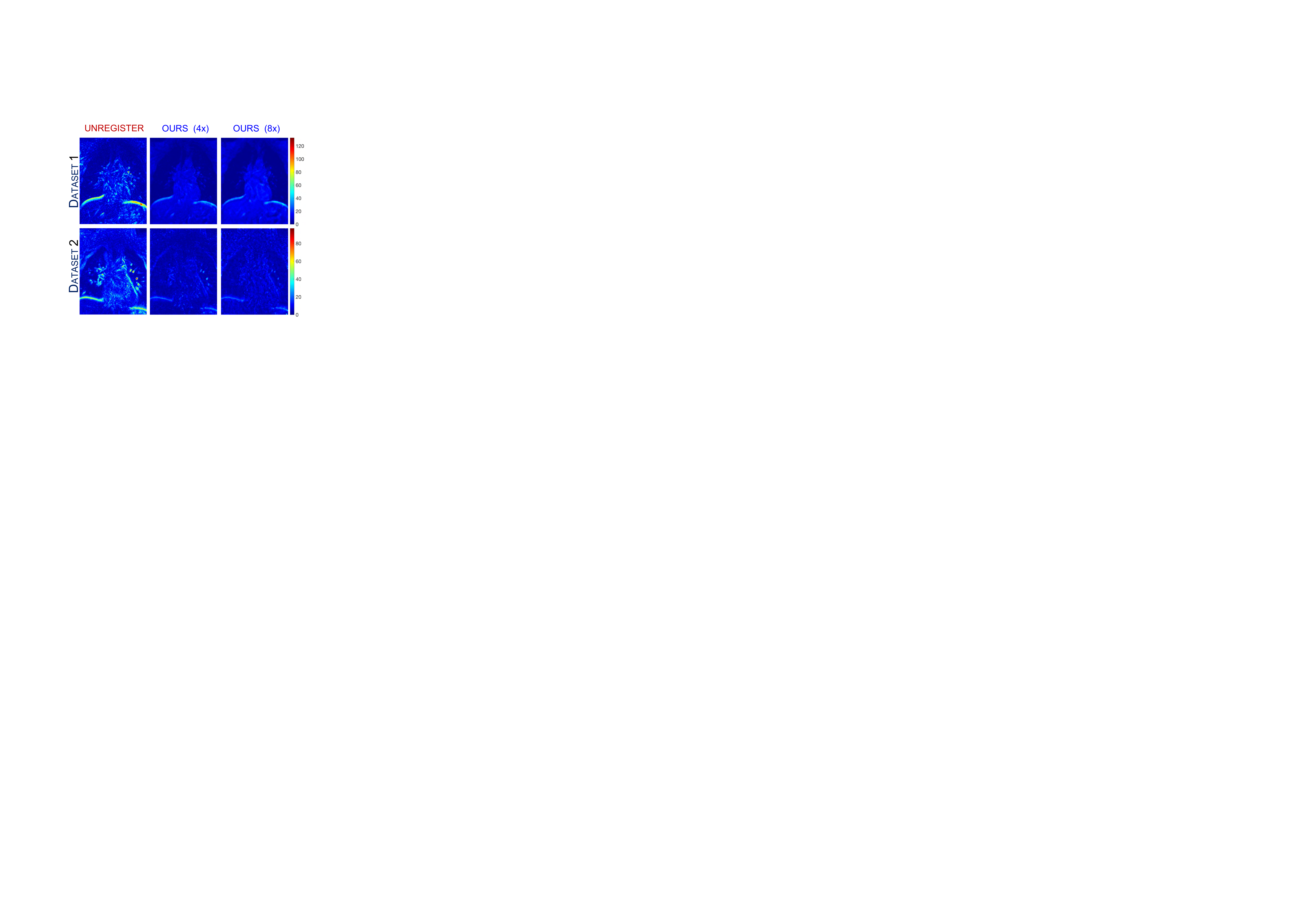}
    \caption{Difference maps. From left to right: average difference maps of the unregistered sequence, average difference map of the corrected sequence for an acceleration factor of  4 and 8 and the colorbar for Datasets 1 and 2.}
    \label{fig:diff_maps}
\end{figure}

\smallskip
\textbf{$\triangleright$ The CPU Cost of Our Multi-Tasking Approach -  Does It Pay Off?} From previous sections, we demonstrated that our approach achieves a better reconstruction in comparison with other approaches, however, does this improvement come to pay off in terms of computational time? Therefore, in this section, we highlight the computational advantages of our model. We remark to the reader that all comparisons were run under the same conditions.

The CPU time, for all approaches, is displayed  at Fig.~\ref{fig:cpu} for several acceleration factors (the remaining can be found in Section IV of the Supplementary Material). Firstly, we observe that, in terms of sequential models, our model outperforms RIGID and HYPER reporting the lower CPU time. However and in terms of the other multi-task methods, the computational time for GW-CS is much longer compared to our method (and only performing two-tasks). Whilst the DC-CS approach readily competes with our approach, from the image quality standpoint our method offers by far better results in terms of reconstruction. We emphasise that the CPU times for our model and DC-CS are still on the same range, but the proposed method is performing three-tasks instead of two like the DC-CS. These advantages highlight our optimisation scheme that allows computing a complex problem in a very computational tractable form.

\section{Conclusion}
In this work, we proposed a novel variational multi-task framework to achieve higher quality and super-resolved reconstructions. Our method compensates for motion due to breathing in undersampled data. To the best of our knowledge, it is the first variational framework that allows computing three tasks jointly. 

In particular, our multi-task framework is composed of four major components: an $L^2$ fidelity term intertwining MRI reconstruction, super-resolution and registration; a weighted TV ensuring robustness of our method to intensity changes by promoting edge alignment; a TV regulariser of the super-resolved reconstruction; and a hyperelasticity-based regulariser modelling biological tissue behavior and allowing for large and smooth deformations. We exploit the temporal redundancy to correct for blurring artefacts and increase image quality. As a result, we obtain a single highly resolved and clear image reconstruction representing the true underlying anatomy. 

The advantages of our model is that we  guarantee preservation of anatomical structures whilst keeping fine details and less blurry and noise artefacts in the final reconstructions. We extensively evaluated our method against sequential and another multi-task methods from the body of literature. We demonstrated that our method achieves the best results whilst demanding low CPU time.  Our method was further supported by a user-study (experts).

\textbf{Future Work.} This multi-task framework is indeed very well-suited for the plug-and-play setting when one (or more) imaging tasks could be replaced by different algorithms. For instance, the modelling of the regularisation functional for the high resolution image reconstruction could be replaced in a plug-and-play fashion. It opens the door to hybrid methods as deep-learning can be used for the super-resolution task.

\section*{Acknowledgment}
VC acknowledges the financial support of the Cambridge Cancer Centre and the Cambridge Research UK. Support from the Centre for Mathematical Imaging in Healthcare (CMIH), Cantab Capital Institute for the Mathematics of Information (CCIMI) and the National Physical Laboratory (NPL) are greatly acknowledged.



\bibliographystyle{IEEEtran}
\bibliography{refs}

\begin{thebibliography}{10}
\providecommand{\url}[1]{#1}
\csname url@samestyle\endcsname
\providecommand{\newblock}{\relax}
\providecommand{\bibinfo}[2]{#2}
\providecommand{\BIBentrySTDinterwordspacing}{\spaceskip=0pt\relax}
\providecommand{\BIBentryALTinterwordstretchfactor}{4}
\providecommand{\BIBentryALTinterwordspacing}{\spaceskip=\fontdimen2\font plus
\BIBentryALTinterwordstretchfactor\fontdimen3\font minus
  \fontdimen4\font\relax}
\providecommand{\BIBforeignlanguage}[2]{{%
\expandafter\ifx\csname l@#1\endcsname\relax
\typeout{** WARNING: IEEEtran.bst: No hyphenation pattern has been}%
\typeout{** loaded for the language `#1'. Using the pattern for}%
\typeout{** the default language instead.}%
\else
\language=\csname l@#1\endcsname
\fi
#2}}
\providecommand{\BIBdecl}{\relax}
\BIBdecl

\bibitem{brown2015mri}
M.~A. Brown, R.~C. Semelka, and B.~M. Dale, \emph{MRI: basic principles and
  applications}.\hskip 1em plus 0.5em minus 0.4em\relax John Wiley \& Sons,
  2015.

\bibitem{zaitsev2015motion}
M.~Zaitsev, J.~Maclaren, and M.~Herbst, ``Motion artifacts in {MRI}: a complex
  problem with many partial solutions,'' \emph{J Magn Reson Imaging,}, vol.~42,
  no.~4, 2015.

\bibitem{sachs1995diminishing}
T.~S. Sachs, C.~H. Meyer, P.~Irarrazabal, B.~S. Hu, D.~G. Nishimura, and
  A.~Macovski, ``The diminishing variance algorithm for real-time reduction of
  motion artifacts in {MRI},'' \emph{Magn Reson Med}, vol.~34, no.~3, pp.
  412--422, 1995.

\bibitem{havsteen2017movement}
I.~Havsteen, A.~Ohlhues, K.~H. Madsen, J.~D. Nybing, H.~Christensen, and
  A.~Christensen, ``Are movement artifacts in magnetic resonance imaging a real
  problem?— {A} narrative review,'' \emph{Frontiers in neurology}, vol.~8, p.
  232, 2017.

\bibitem{Andre2015}
J.~B. Andre, B.~W. Bresnahan, M.~Mossa-Basha, M.~N. Hoff, C.~P. Smith,
  Y.~Anzai, and W.~A. Cohen, ``Toward quantifying the prevalence, severity, and
  cost associated with patient motion during clinical mr examinations,''
  \emph{Journal of the American College of Radiology}, vol.~12, no.~7, pp.
  689--695, 2015.

\bibitem{Birn2004}
R.~M. Birn, R.~W. Cox, and P.~A. Bandettini, ``Experimental designs and
  processing strategies for fmri studies involving overt verbal responses,''
  \emph{Neuroimage}, vol.~23, no.~3, pp. 1046--1058, 2004.

\bibitem{Setser2000}
R.~M. Setser, S.~E. Fischer, and C.~H. Lorenz, ``Quantification of left
  ventricular function with magnetic resonance images acquired in real time,''
  \emph{J Magn Reson Imaging}, vol.~12, no.~3, pp. 430--438, 2000.

\bibitem{Plein2001}
S.~Plein, W.~H. Smith, J.~P. Ridgway, A.~Kassner, D.~J. Beacock, T.~N. Bloomer,
  and M.~U. Sivananthan, ``Qualitative and quantitative analysis of regional
  left ventricular wall dynamics using real-time magnetic resonance imaging:
  comparison with conventional breath-hold gradient echo acquisition in
  volunteers and patients,'' \emph{J Magn Reson Imaging}.

\bibitem{Kaji2001}
S.~Kaji, P.~C. Yang, A.~B. Kerr, W.~W. Tang, C.~H. Meyer, A.~Macovski, J.~M.
  Pauly, D.~G. Nishimura, and B.~S. Hu, ``Rapid evaluation of left ventricular
  volume and mass without breath-holding using real-time interactive cardiac
  magnetic resonance imaging system,'' \emph{Journal of the American College of
  Cardiology}, vol.~38, no.~2, pp. 527--533, 2001.

\bibitem{GEORGE2006924}
``Audio-visual biofeedback for respiratory-gated radiotherapy: Impact of audio
  instruction and audio-visual biofeedback on respiratory-gated radiotherapy,''
  \emph{International Journal of Radiation Oncology*Biology*Physics}, vol.~65,
  no.~3, pp. 924 -- 933, 2006.

\bibitem{JIANG2006141}
``Technical aspects of image-guided respiration-gated radiation therapy,''
  \emph{Medical Dosimetry}, vol.~31, no.~2, pp. 141 -- 151, 2006.

\bibitem{Keall_2000}
P.~J. Keall, V.~R. Kini, S.~S. Vedam, and R.~Mohan, ``Motion adaptive x-ray
  therapy: a feasibility study,'' \emph{Physics in Medicine and Biology},
  vol.~46, no.~1, pp. 1--10, nov 2000.

\bibitem{HOISAK2006339}
``Prediction of lung tumour position based on spirometry and on abdominal
  displacement: Accuracy and reproducibility,'' \emph{Radiotherapy and
  Oncology}, vol.~78, no.~3, pp. 339 -- 346, 2006.

\bibitem{lustig2007sparse}
M.~Lustig, D.~Donoho, and J.~M. Pauly, ``Sparse {MRI}: The application of
  compressed sensing for rapid {MR} imaging,'' \emph{Magn Reson Med}, pp.
  1182--1195, 2007.

\bibitem{liang2007spatiotemporal}
Z.-P. Liang, ``Spatiotemporal imagingwith partially separable functions,'' in
  \emph{IEEE International Symposium on Biomedical Imaging: From Nano to
  Macro}, 2007, pp. 988--991.

\bibitem{Lingala::2011}
S.~G. Lingala, Y.~Hu, E.~DiBella, and M.~Jacob, ``Accelerated dynamic mri
  exploiting sparsity and low-rank structure: kt slr,'' \emph{IEEE Trans Med
  Imaging}, pp. 1042--1054, 2011.

\bibitem{lingala2011accelerated}
------, ``Accelerated dynamic mri exploiting sparsity and low-rank structure:
  kt slr,'' \emph{IEEE Trans Med Imaging}, vol.~30, no.~5, pp. 1042--1054,
  2011.

\bibitem{Otazo::2015}
R.~Otazo, E.~Cand{\`e}s, and D.~K. Sodickson, ``Low-rank plus sparse matrix
  decomposition for accelerated dynamic mri with separation of background and
  dynamic components,'' \emph{Magn Reson Med}, pp. 1125--1136, 2015.

\bibitem{zhang2015accelerating}
T.~Zhang, J.~M. Pauly, and I.~R. Levesque, ``Accelerating parameter mapping
  with a locally low rank constraint,'' \emph{Magn Reson Med}, vol.~73, no.~2,
  pp. 655--661, 2015.

\bibitem{Lingala::2015}
S.~G. {Lingala}, E.~{DiBella}, and M.~{Jacob}, ``Deformation corrected
  compressed sensing (dc-cs): A novel framework for accelerated dynamic mri,''
  \emph{IEEE Trans Med Imaging}, vol.~34, no.~1, pp. 72--85, Jan 2015.

\bibitem{Royuela-del-Val2016}
J.~Royuela-del Val, L.~Cordero-Grande, F.~Simmross-Wattenberg,
  M.~Mart{\'\i}n-Fern{\'a}ndez, and C.~Alberola-L{\'o}pez, ``Nonrigid groupwise
  registration for motion estimation and compensation in compressed sensing
  reconstruction of breath-hold cardiac cine mri,'' \emph{Magn Reson Med},
  vol.~75, no.~4, pp. 1525--1536, 2016.

\bibitem{aviles2018compressed}
A.~I. Aviles-Rivero, G.~Williams, M.~J. Graves, and C.-B. Sch\"onlieb,
  ``Compressed sensing plus motion ({CS+M}): A new perspective for improving
  undersampled {MR} image reconstruction,'' \emph{arXiv preprint
  arXiv:1810.10828}, 2018.

\bibitem{odille2016joint}
F.~Odille, A.~Menini, J.-M. Escany{\'e}, P.-A. Vuissoz, P.-Y. Marie,
  M.~Beaumont, and J.~Felblinger, ``Joint reconstruction of multiple images and
  motion in {MRI}: Application to free-breathing myocardial {T2}
  quantification,'' \emph{IEEE Trans Med Imaging}, vol.~35, no.~1, pp.
  197--207, 2016.

\bibitem{blume2010joint}
M.~Blume, A.~Martinez-Moller, A.~Keil, N.~Navab, and M.~Rafecas, ``Joint
  reconstruction of image and motion in gated positron emission tomography,''
  \emph{IEEE Trans Med Imaging}, vol.~29, no.~11, pp. 1892--1906, 2010.

\bibitem{jacobson2003joint}
M.~Jacobson and J.~Fessler, ``Joint estimation of image and deformation
  parameters in motion-corrected {PET},'' in \emph{IEEE NSS Conf}, vol.~5,
  2003, pp. 3290--3294.

\bibitem{gupta2003fast}
S.~N. Gupta, M.~Solaiyappan, G.~M. Beache, A.~E. Arai, and T.~K. Foo, ``Fast
  method for correcting image misregistration due to organ motion in
  time-series {MRI} data,'' \emph{Magn Reson Med}, vol.~49, no.~3.

\bibitem{adluru2006model}
G.~Adluru, E.~V. DiBella, and M.~C. Schabel, ``Model-based registration for
  dynamic cardiac perfusion {MRI},'' \emph{J Magn Reson Imaging.
  2006;24(5):1062-70}.

\bibitem{wong2008first}
K.~K. Wong, E.~S. Yang, E.~X. Wu, H.-F. Tse, and S.~T. Wong, ``First-pass
  myocardial perfusion image registration by maximization of normalized mutual
  information,'' \emph{JMRI}, vol.~27, no.~3, pp. 529--537, 2008.

\bibitem{johansson2018rigid}
A.~Johansson, J.~Balter, and Y.~Cao, ``Rigid-body motion correction of the
  liver in image reconstruction for golden-angle stack-of-stars {DCE} {MRI},''
  \emph{Magn Reson Med}, vol.~79, no.~3, pp. 1345--1353, 2018.

\bibitem{LedesmaCarbayoKellman2007}
H.~L. A. A. M.~E. Ledesma-Carbayo~MJ, Kellman~P, ``Motion corrected
  free-breathing delayed-enhancement imaging of myocardial infarction using
  nonrigid registration.'' \emph{J Magn Reson Imaging}, vol.~26, no.~1, pp.
  184--90., Jul. 2007.

\bibitem{ledesma2007motion}
M.~J. Ledesma-Carbayo, P.~Kellman, A.~E. Arai, and E.~R. McVeigh, ``Motion
  corrected free-breathing delayed-enhancement imaging of myocardial infarction
  using nonrigid registration,'' \emph{J Magn Reson Imaging.
  2007;26(1):184-90}, 2007.

\bibitem{li2015expiration}
Z.~Li, J.~A. Tielbeek, M.~W. Caan, C.~A. Puylaert, M.~L. Ziech, C.~Y. Nio,
  J.~Stoker, L.~J. van Vliet, and F.~M. Vos, ``Expiration-phase template-based
  motion correction of free-breathing abdominal dynamic contrast enhanced
  {MRI},'' \emph{IEEE TBE}, vol.~62, no.~4, pp. 1215--1225, 2015.

\bibitem{jansen2017evaluation}
M.~Jansen, H.~Kuijf, W.~Veldhuis, F.~Wessels, M.~Van~Leeuwen, and J.~Pluim,
  ``Evaluation of motion correction for clinical dynamic contrast enhanced
  {MRI} of the liver,'' \emph{Physics in Medicine \& Biology}, vol.~62, no.~19,
  p. 7556, 2017.

\bibitem{burger2018variational}
M.~Burger, H.~Dirks, and C.-B. Sch\"onlieb, ``A variational model for joint
  motion estimation and image reconstruction,'' \emph{SIAM J. Imaging Sciences,
  11, 1, pp. 94-128}, 2018.

\bibitem{mair2006estimation}
B.~A. Mair, D.~R. Gilland, and J.~Sun, ``Estimation of images and nonrigid
  deformations in gated emission {CT},'' \emph{IEEE Trans Med Imaging},
  vol.~25, no.~9, pp. 1130--1144, 2006.

\bibitem{schumacher2009combined}
H.~Schumacher, J.~Modersitzki, and B.~Fischer, ``Combined reconstruction and
  motion correction in {SPECT} imaging,'' \emph{ITNS 56, (1)}.

\bibitem{fessler2009}
S.~Chun and J.~Fessler, ``Joint image reconstruction and nonrigid motion
  estimation with a simple penalty that encourages local invertibility,''
  \emph{Proc SPIE}, vol. 7258.

\bibitem{fessler2010optimization}
J.~A. Fessler, ``Optimization transfer approach to joint
  registration/reconstruction for motion-compensated image reconstruction,''
  \emph{ISBI}, 2010.

\bibitem{yang2017quicksilver}
X.~Yang, R.~Kwitt, M.~Styner, and M.~Niethammer, ``Quicksilver: Fast predictive
  image registration--a deep learning approach,'' \emph{NeuroImage}, vol. 158,
  pp. 378--396, 2017.

\bibitem{de2017end}
B.~D. de~Vos, F.~F. Berendsen, M.~A. Viergever, M.~Staring, and I.~I{\v{s}}gum,
  ``End-to-end unsupervised deformable image registration with a convolutional
  neural network,'' in \emph{DLMIA}, 2017, pp. 204--212.

\bibitem{ROF}
L.~I. Rudin, S.~Osher, and E.~Fatemi, ``{N}onlinear total variation based noise
  removal algorithms.'' \emph{Physica D, vol. 60, pp. 259-268}, 1992.

\bibitem{Majumdar::2012}
A.~Majumdar and R.~K. Ward, ``Exploiting rank deficiency and transform domain
  sparsity for mr image reconstruction,'' \emph{Magn Reson Imaging}, pp. 9--18,
  2012.

\bibitem{ball_1981}
J.~M. Ball, ``Global invertibility of sobolev functions and the
  interpenetration of matter,'' \emph{Proceedings of the Royal Society of
  Edinburgh: Section A Mathematics}, vol.~88, no. 3-4, p. 315–328, 1981.

\bibitem{Corona}
V.~Corona, A.~Aviles-Rivero, N.~Debroux, M.Grave, C.~L. Guyader, C.-B.
  Schoenlieb, and G.~Williams, ``Multi-tasking to correct: motion-compensated
  {MRI} via joint reconstruction and registration,'' \emph{Scale Space and
  Variational Methods in Computer Vision LNCS conference proceedings,
  Springer.}, 2019.

\bibitem{Chambolle2004}
A.~Chambolle, ``An algorithm for total variation minimization and
  applications,'' \emph{J Math Imaging Vis}, vol.~20, no.~1, pp. 89--97, 2004.

\bibitem{pd}
A.~Chambolle and T.~Pock, ``{A} first-order primal-dual algorithm for convex
  problems with applications to imaging,'' \emph{J Math Imaging Vis 40: 120}.

\bibitem{christensen}
G.~E. {Christensen}, R.~D. {Rabbitt}, and M.~I. {Miller}, ``Deformable
  templates using large deformation kinematics,'' \emph{IEEE Trans Image
  Process}, vol.~5, no.~10, pp. 1435--1447, Oct 1996.

\bibitem{BAUMGARTNER201783}
C.~F. Baumgartner, C.~Kolbitsch, J.~R. McClelland, D.~Rueckert, and A.~P. King,
  ``Autoadaptive motion modelling for {MR}-based respiratory motion
  estimation,'' \emph{Med Image Anal}, vol.~35, pp. 83 -- 100, 2017.

\bibitem{Siebenthal2007}
M.~von Siebenthal, G.~Szekely, U.~Gamper, P.~Boesiger, A.~Lomax, and P.~Cattin,
  ``{4D} {MR} imaging of respiratory organ motion and its variability,''
  \emph{Physics in Medicine \& Biology}, vol.~52, 2007.

\bibitem{2009-FAIR}
J.~Modersitzki, \emph{{FAIR}: Flexible Algorithms for Image
  Registration}.\hskip 1em plus 0.5em minus 0.4em\relax Philadelphia: SIAM,
  2009.

\bibitem{debroux-bib:baldi}
A.~Baldi, ``Weighted {B}{V} functions,'' \emph{Houston J. Math.}, vol.~27,
  no.~3, pp. 683--705, 2001.

\bibitem{debroux-bib:dacorogna}
B.~Dacorogna, \emph{Direct {M}ethods in the {C}alculus of {V}ariations,
  {S}econd {E}dition}.\hskip 1em plus 0.5em minus 0.4em\relax Springer, 2008.

\bibitem{corona-bib:ball}
``Global invertibility of {S}obolev functions and the interpenetration of
  matter,'' vol.~88.

\bibitem{debroux-bib:brezis}
H.~Br{e}zis, \emph{Analyse fonctionelle:}, ser. Collection Math{\'e}matiques
  appliqu{\'e}es pour la ma{\^\i}trise, 1983.

\bibitem{corona-bib:demengel}
F.~Demengel, G.~Demengel, and R.~Ern{\'e}, \emph{Functional Spaces for the
  Theory of Elliptic Partial Differential Equations}, ser. Universitext.\hskip
  1em plus 0.5em minus 0.4em\relax Springer London, 2012.

\end{thebibliography}


\begin{table*}[h!]
\begin{tabular}{c}
\LARGE \textsc{Supplemental Material for the Article: }                                                                                                                                                                                                                                 \\ \\ \\
\begin{tabular}[c]{@{}c@{}} \Huge Variational Multi-Task MRI Reconstruction: Joint  \\    \Huge Reconstruction, Registration and Super-Resolution\end{tabular}                                                                                            \\ \\
\begin{tabular}[c]{@{}c@{}} \large Veronica Corona, Angelica I. Aviles-Rivero, No\'emie Debroux, Carole Le Guyader, \\  \large Carola-Bibiane Sch\"onlieb\end{tabular}
\end{tabular}
\end{table*}

\setcounter{section}{0}
\section{Outline}
This supplementary material extends the experimental results, practicalities and theoretical analysis to further support our proposed multi-task framework.  This document is organised as follows. 
\begin{itemize}
    \item \textbf{Section II}. In this section, we offer a detailed explanation about the influence of the parameters of out variational multi-task model.
    \item \textbf{Section III}. We further describe the user-study protocol that we follow to support our outcomes and comparisons.  
   \textbf{ \item \textbf{Section IV}}. We display further visual comparisons of our model against some works from the state-of-the-art. This section further validate our claim in terms of our model performance. 
    \item \textbf{Section V}. We provide further mathematical details for the paper. First, we recall the definition of the weighted total variation, then we introduce the proof, of Theorem III.1 regarding the existence of minimisers result. 
\end{itemize}

\section{Parameters Reasoning of Our Multi-Task Framework} 
In this section, we discuss the influence of each parameter. $a_1$ and $a_2$ control the regularisation of the deformations. Whilst the former acts on the smoothness of the deformations, the latter can be seen as a measure of rigidity. That is- the bigger $a_2$ is, the more rigid the deformations are and the less accurate the registration becomes. It thus behaves as a trade-off between the ability to handle large and nonlinear deformations, and topology preservation.

Moreover, $\gamma_1$ and $\gamma_2$ are chosen big, to ensure the closeness of the auxiliary variables to the original ones, and $\theta$ is set small for the same purpose. $\gamma_3$ weights the fidelity term joining the three tasks, and it is often chosen to be close to 1. Finally, $\alpha$ offers a balance between regularity and fidelity to the data for the super-resolved reconstructed image $u$.  

Finally, we also explicitly define he regridding algorithm that we used in Section III.D from the main paper. This is displayed at Algorithm~\ref{alg::regridding}.

\begin{figure}[t!]
    \centering
    \includegraphics[width=0.45\textwidth]{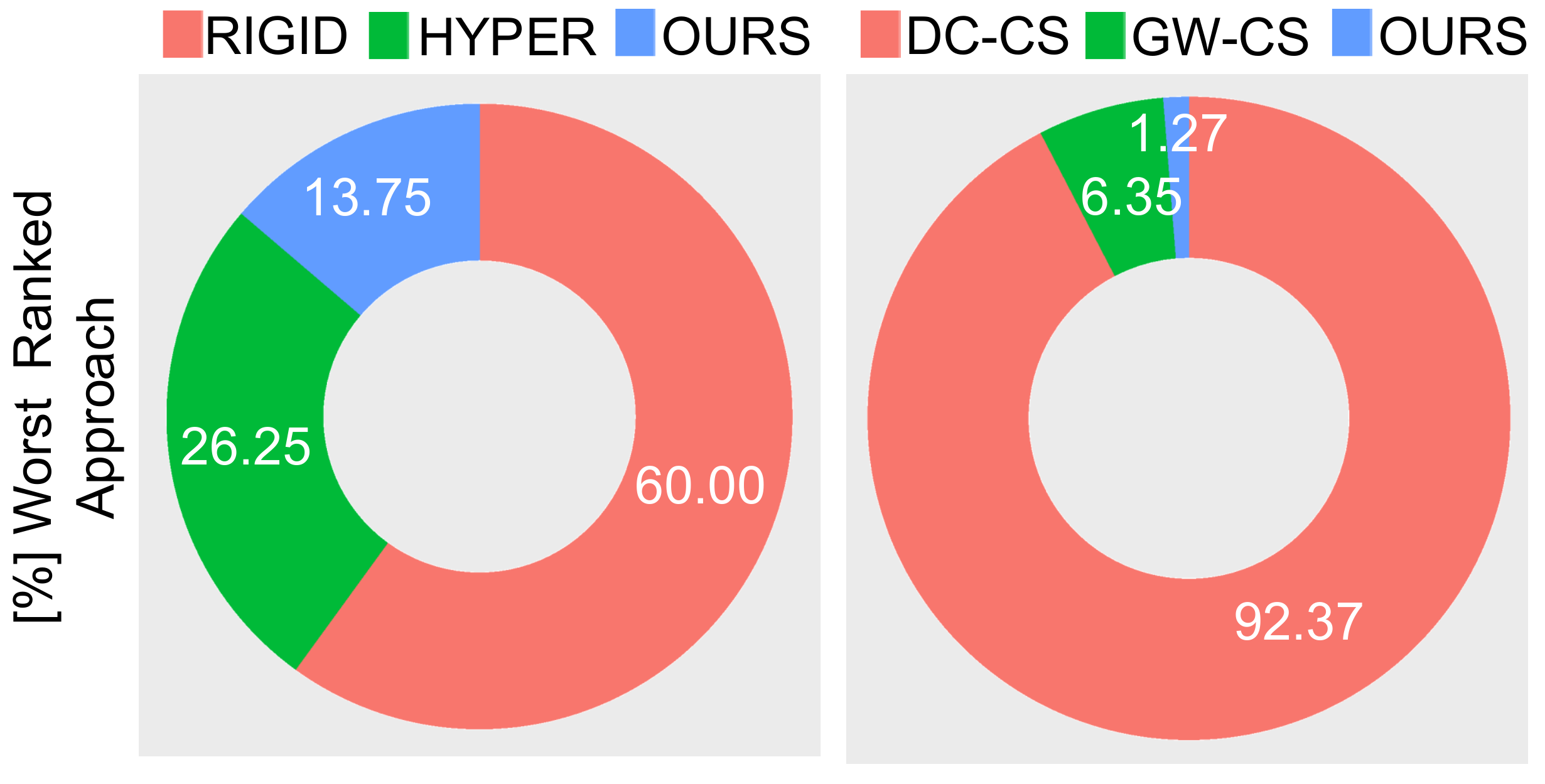}
    \caption{Plots displays the percentage of responses supporting the worst ranked approach.}
    \label{fig:userstudy2}
\end{figure}
\begin{algorithm}[t!]
\begin{algorithmic}[1]
\State Initialisation $z^0=0$, $\phi^0=\mbox{Id}$, $regrid\_count=0$.
\For{$n=1,\cdots,N$}
\State Update $z^n$ and $\phi^n$.
\If{$\mbox{det}\,\nabla\phi^n<tol$}
\State $regrid\_count=regrid\_count+1$.
\State $h=h\circ \phi^{n-1}$.
\State Save $tab\_phi(regrid\_count)=\phi^{n-1}$.
\State $\phi^{n}=\mbox{Id}$, $z^{n}=0$.
\EndIf
\EndFor
\If{$regrid\_count>0$}
\State $\phi^{final}=tab\_phi(1)\circ \cdots \circ tab\_phi(regrid\_count)$.
\EndIf
\end{algorithmic}
\caption{Regridding algorithm}
\label{alg::regridding}
\end{algorithm}

\section{Further Details on the User-Study}
To further support our results, we performed a user-study following standard protocol in clinical settings. That is- to ask experts and trainees to evaluate the reconstructions based on a scoring system. For this study, we had twelve experts and trainees from the area of radiology.

\begin{figure*}[t!]
    \centering
    \includegraphics[width=1\textwidth]{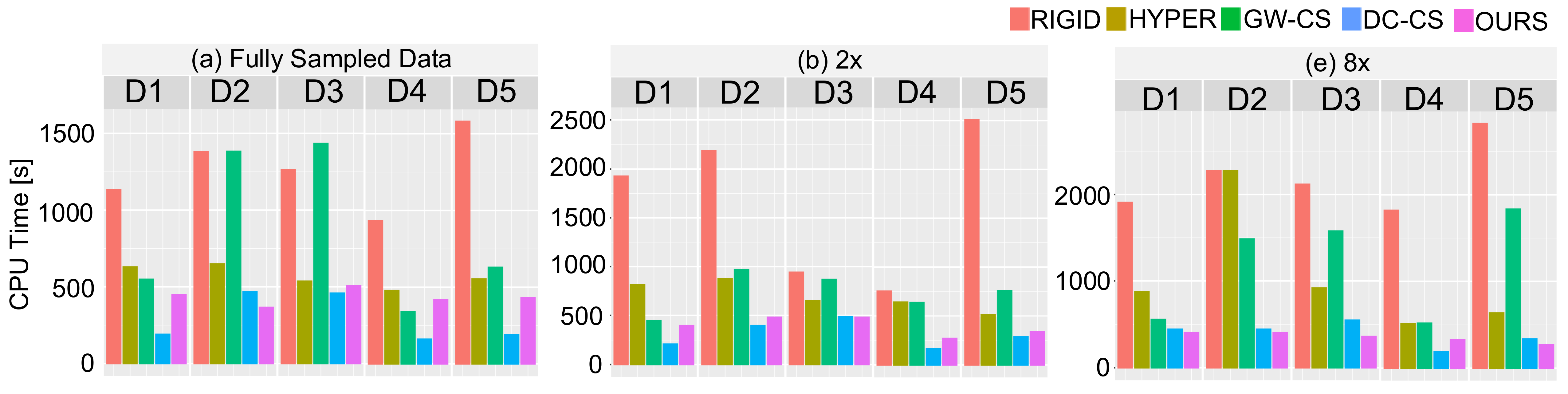}
    \caption{Computational performance comparison between sequential (three tasks), joint (two tasks) and our approach. Elapsed time in seconds. The sequential approaches are definitely much slower than our proposed method. We can see that our approach is comparable and competitive with joint approaches although slightly slower than the DC-CS, which, however, only computes two tasks..}
    \label{fig:cpu2}
\end{figure*}

\begin{figure}[t!]
    \centering
    \includegraphics[width=0.5\textwidth]{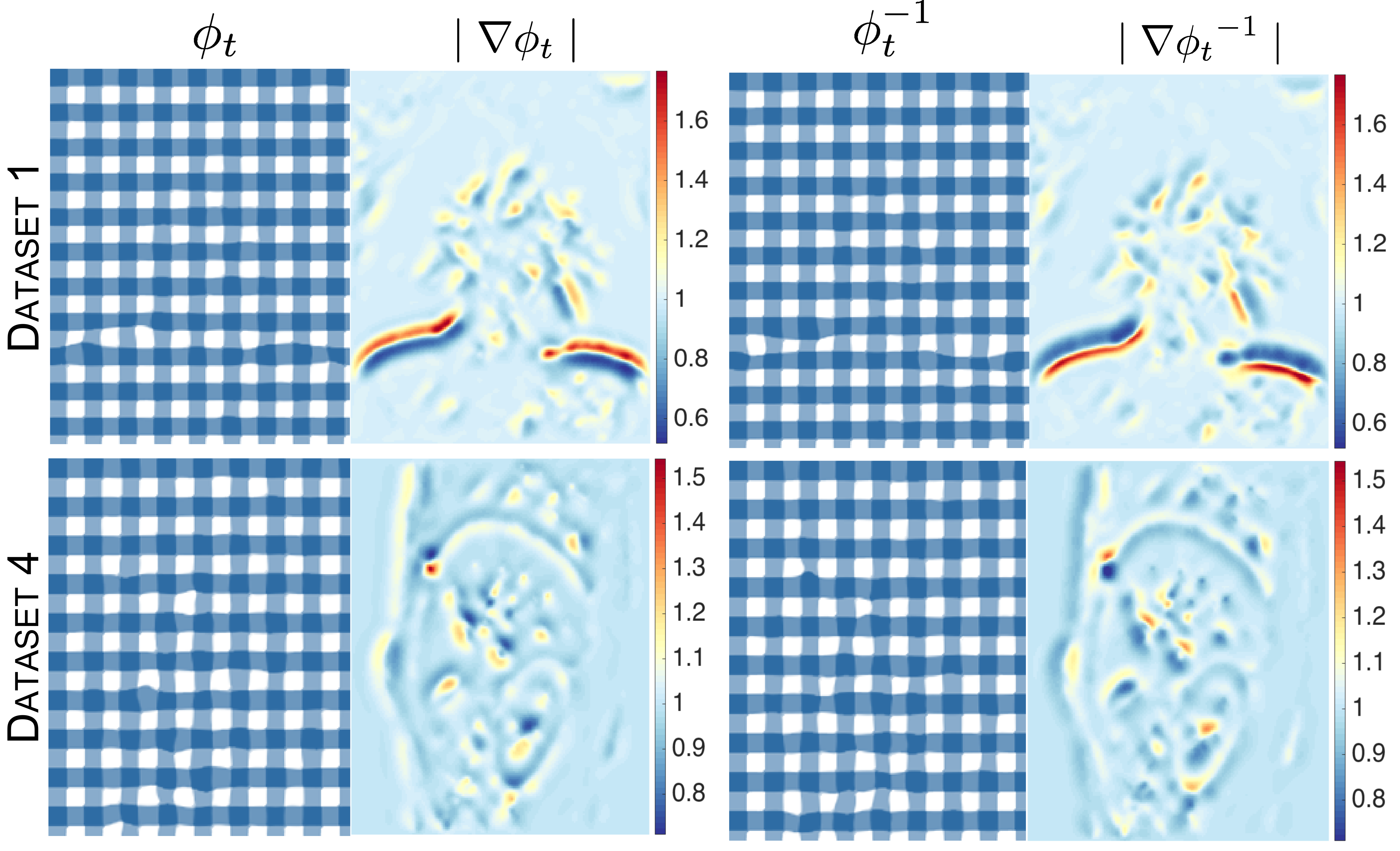}
    \caption{Estimated motion and determinant maps of the deformation Jacobian. This is
shown for the transformation $\phi_i$ and its inverse $\phi_i^{-1}$ for two datasets (1 and 4).}
    \label{fig:detMaps}
\end{figure}

\textbf{Scoring Procedure.} We create an electronic survey in which, after giving the instructions to the users, they were provided with two-part evaluation. For the first part, they evaluated the reconstructions related to sequential models. To do this, they were provided with the reconstructions of the five dataset, for acceleration factor =\{fully sampled, 2x, 4x, 6x, 8x\}, and for the approaches  =\{RIGID, HYPER, OURS\} (see main paper for description on these models). With the purpose of capturing they scores, we design a three-point Likert rating scale in which experts were asked to indicate the level of agreement, ranging from best reconstruction to worst reconstruction. Similarly, this protocol was followed to evaluate other multi-task methods = \{DC-CD, GW-CS, OURS\}. Fig.~\ref{fig:userstudy2} displays the results in terms of \% the worst ranked for all compared approaches. 


\textbf{Statistical Analysis.} The circle plots displayed in the main paper and in Fig.~\ref{fig:userstudy2} of this supplementary material, reflects the averaged results of all scoring (i.e. for all acceleration and all reconstructions). However, to test if there was statistical significant difference between the approaches we took into account the the scoring per each reconstruction, and ran the nonparametric Friedman test to compare the three approaches (for both sequential and multi-task frameworks). We also applied corrections for multiple comparisons during the statistical analysis. Overall, we found statistical significant difference between our approach and the compared ones. This supports our previous discussions on our model offering the best reconstruction in comparison to the compared approached for both sequential and joint models.

\section{Further Visual Experiments}
In this section, we extend the reconstruction comparisons from the main paper. In particular, for the sequential model we now included Dataset 2, 4 and 5. The results are displayed in Figs.~\ref{fig:sD2B}, \ref{fig:sD4B} and \ref{fig:sD5B} in which we observe that our joint approach successfully corrects for motion and aligns the different acquisitions resulting in a sharp reconstruction that contains very fine details, such as preserving the texture in the heart and vessels in the liver and kidney. It further support the finding from the main paper.

Moreover, we also provide further evaluation of our framework against other multi-task approaches (DC-CS~\cite{Lingala::2015} and  GW-CS~\cite{Royuela-del-Val2016}). The reconstruction results are displayed at Figs.~\ref{fig:jD2}, \ref{fig:jD4} and \ref{fig:jD5}. In a closer look at those reconstructions, one can observe that our approach produces very sharp reconstructions while preserving the texture in the relevant anatomical areas. This is validated for different acceleration factors. In contrast, DC-CS fails to carefully align the acquisition as we can see from the blurring artefacts around the edges. Additionally all the small details get blurred and smeared out. The GW-CS seems to perform reasonably well for fully-sampled data, however failing to remove the noise. This is accentuated at higher reduction factors, where the reconstructions contain significant noise which compromises the clinical interpretation of the images. 

Moreover, we present for two datasets the estimated motion $\phi_t$ and its inverse $\phi_t^{-1}$ for a given time frame in Fig.~\ref{fig:detMaps}. We can see that our proposed approach produces a reasonable estimation of the motion, where the
motion fields are visualised by a deformation grid. Additionally, we show the corresponding Jacobian determinant maps for the deformations. In these plots, we can see that the determinants remain positive meaning that our model ensures topology preservation both from a mathematical and practical point of view.
The  values are interpreted as follows: small deformations when values are closer to 1, big
expansions when values are greater than 1, and big contractions when values are
smaller than 1. Moreover, one can observe that the determinants remain positive in all cases, that is to say, our estimated deformations are physically meaningful and preserve the topology as required in a registration framework.

Finally, we also report in Fig.~\ref{fig:cpu2}, the CPU times for the remaining accelaretion factors. Again, we can see that our approach allows computing a complex problem in a very computationally tractable way. It is indeed much faster than sequential approaches and than the GW-CS. It is comparable to the DC-CS approach, which, however, only computes two tasks. 
\begin{figure}
    \centering
    \hspace{-0.5cm}
    \includegraphics[width=0.5\textwidth]{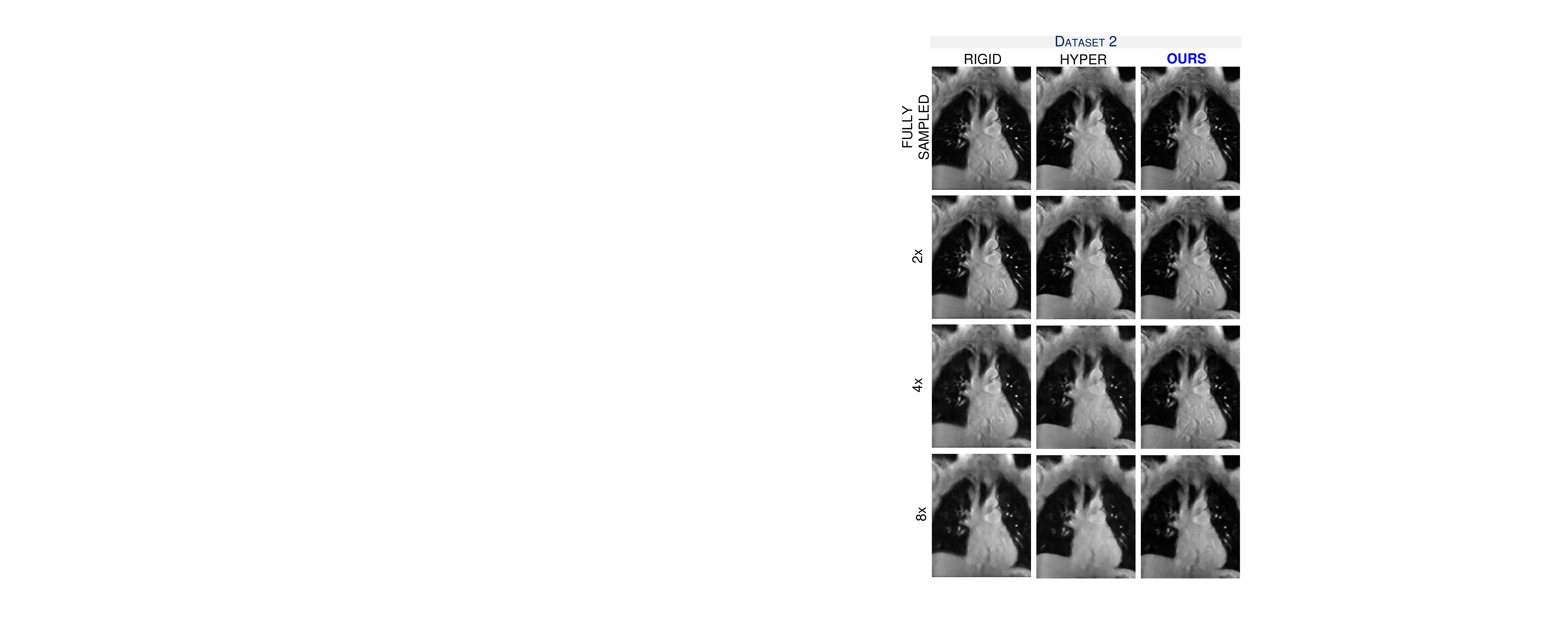}
        \vspace{-0.2cm}
    \caption{Reconstruction results for Dataset 2 compared to sequential approaches based on rigid and hyperelastic registration, for different acceleration factors.}
    \label{fig:sD2B}
\end{figure}
\begin{figure}[t!]
    \centering
    \includegraphics[width=0.5\textwidth]{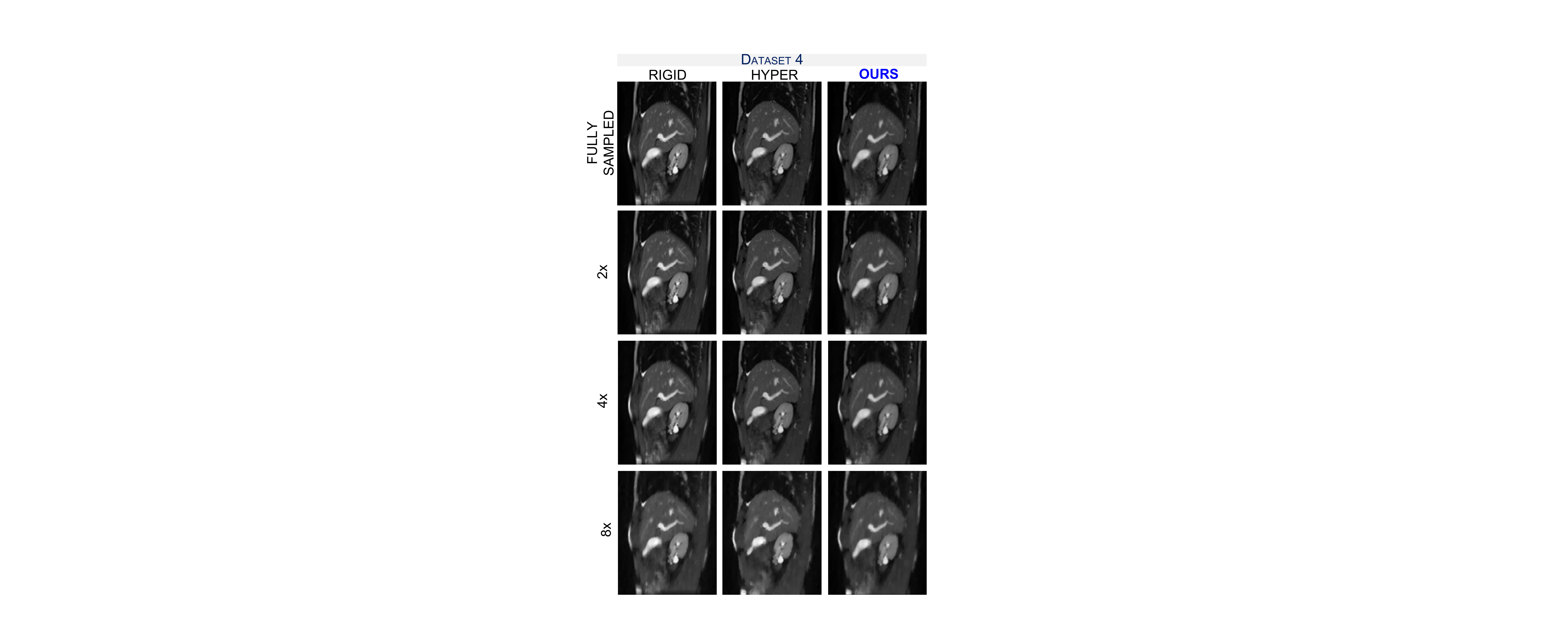}
    \vspace{-0.5cm}
    \caption{Reconstruction results for Dataset 4 compared to sequential approaches based on rigid and hyperelastic registration, for different acceleration factors.}
   \label{fig:sD4B}
\end{figure}

\begin{figure}[t!]
    \centering
    \includegraphics[width=0.5\textwidth, height=0.78\textwidth]{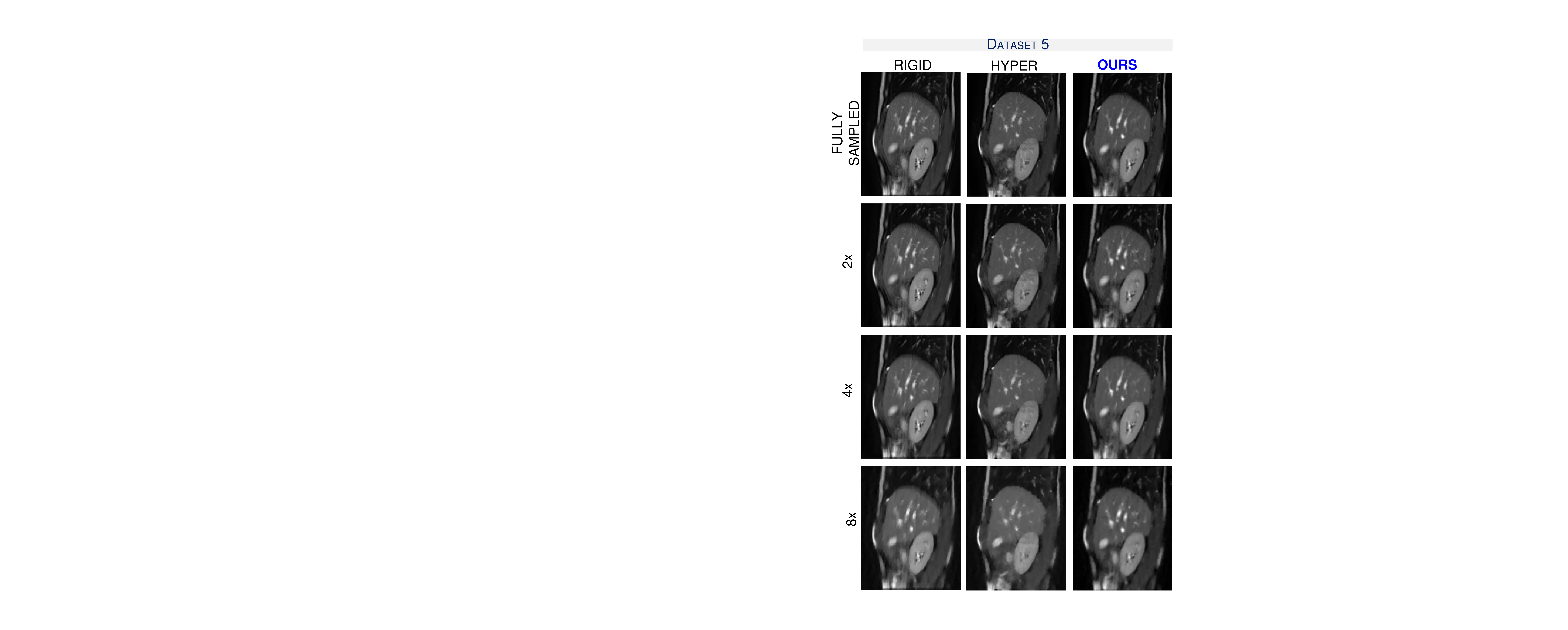}
    \caption{Reconstruction results for Dataset 5 compared to sequential approaches based on rigid and hyperelastic registration, for different acceleration factors. Our reconstructions show sharper edges and finer details especially for higher acceleration factors. Additionally, we can also notice a better contrast preservation in our results.}
   \label{fig:sD5B}
\end{figure}

\begin{figure}
    \centering
    \hspace{-0.5cm}
    \includegraphics[width=0.5\textwidth]{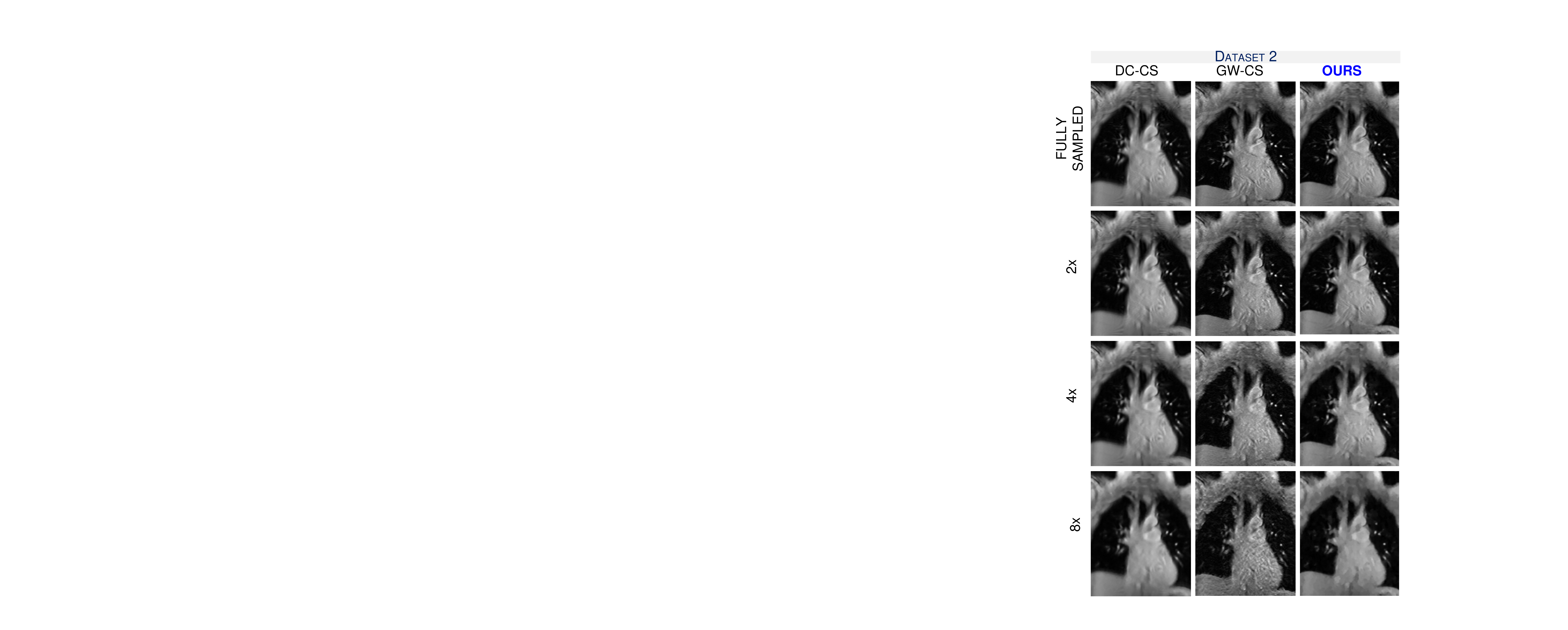}
    \caption{Reconstruction results for Dataset 2 for different acceleration factors and different joint approaches in comparison to our proposed method. We can clearly see that our approach provides the best results in terms of sharp structures and fine texture, while DC-CS results very blurry and GW-CS very noisy. This is particularly accentuated for high undersampling factors.}
    \label{fig:jD2}
\end{figure}

\begin{figure}
    \centering
    \hspace{-0.5cm}
    \includegraphics[width=0.5\textwidth]{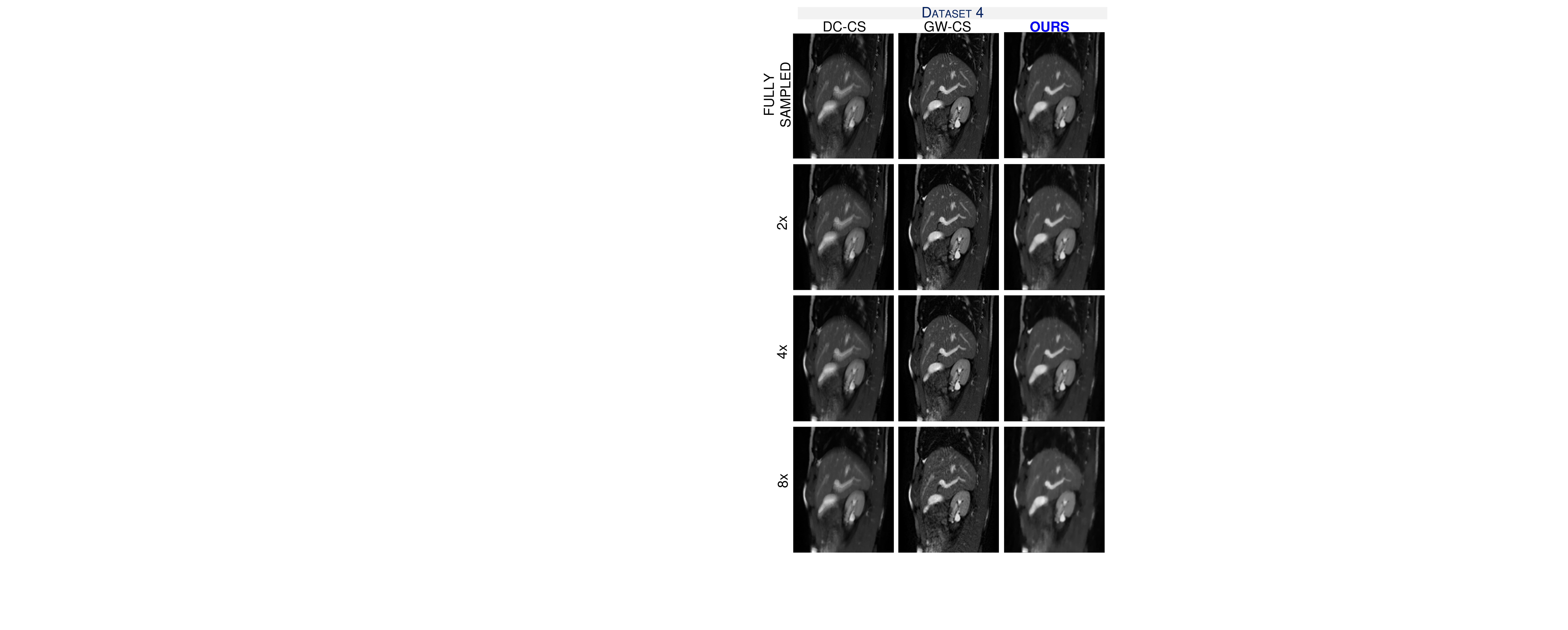}
    \caption{Reconstruction results for Dataset 4 for different acceleration factors and different joint approaches in comparison to our proposed method. We can clearly see that our approach provides the best results in terms of sharp structures and fine texture, while DC-CS results very blurry and GW-CS very noisy. This is particularly accentuated for high undersampling factors.}
    \label{fig:jD4}
\end{figure}

\begin{figure}
    \centering
    \hspace{-0.5cm}
    \includegraphics[width=0.5\textwidth,height=14.32cm]{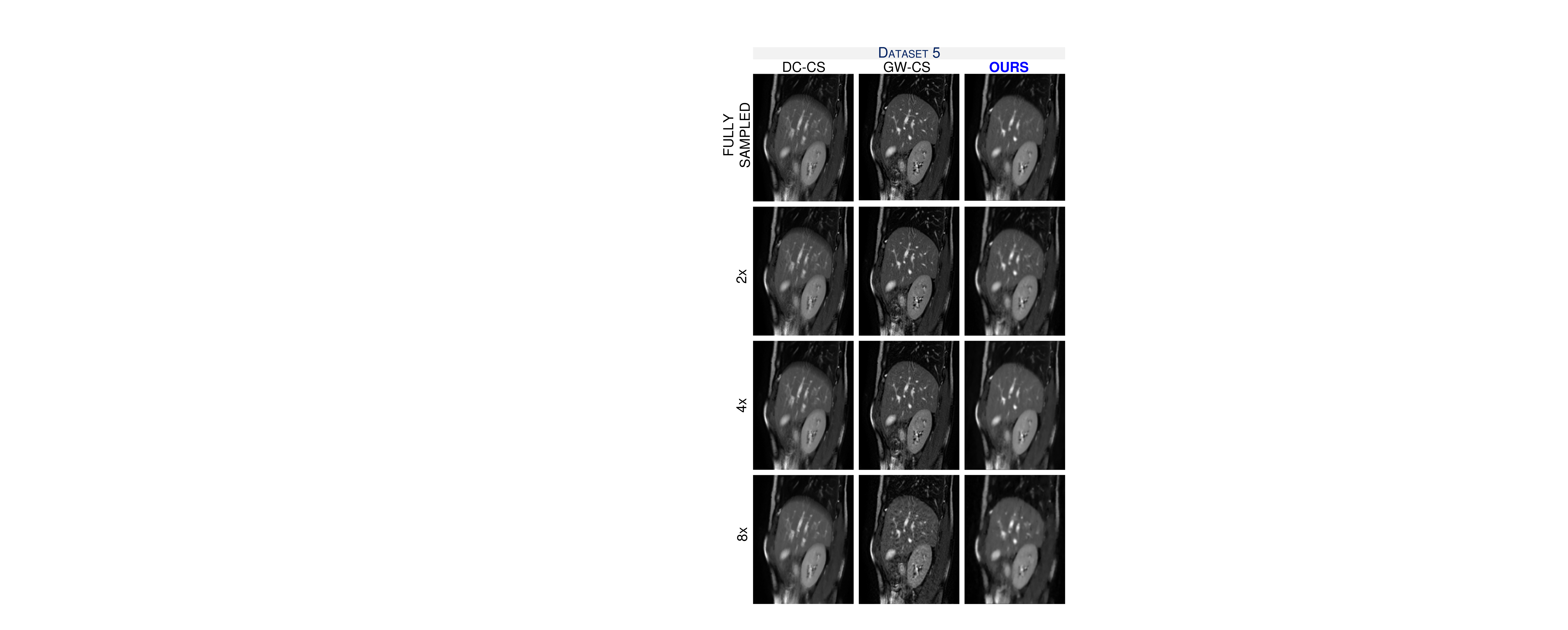}
    \caption{Reconstruction results for Dataset 4 for different acceleration factors and different joint approaches in comparison to our proposed method. We can clearly see that our approach provides the best results in terms of sharp structures and fine texture, while DC-CS results very blurry and GW-CS very noisy. This is particularly accentuated for high undersampling factors.}
    \label{fig:jD5}
\end{figure}




\section{Further Mathematical Details}
\label{sec:proof}
We recall the definition of the weighted BV-space and the associated weighted TV. 

\begin{definition}[\hspace{-0.1cm}{\cite[Definition 2]{debroux-bib:baldi}}]
Let $w$ be a weight function satisfying some properties (defined in \cite{debroux-bib:baldi} and fulfilled by $g_t$). We denote by $BV_w(\Omega)$ the set of functions $f\in L^1(\Omega,w)$, which are integrable with respect to the measure $w(x)dx$, such that:
$
\sup\left\{  \int_\Omega f \mbox{div}(\varphi)\,dx\,:\, |\varphi|\leq w \text{ everywhere},\, \\
 \varphi\in Lip_0(\Omega,\mathbb{R}^2) \right\}$ $<+\infty,
$
\noindent
where $Lip_0(\Omega,\mathbb{R}^2)$ is the space of Lipschitz continuous functions with compact support. We denote by $TV_w$ the previous quantity.
\end{definition}
\subsection{Detailed Proof of Theorem III.1}
We report the detailed proof of the Theorem III.1, from the main paper , regarding the existence of minimisers. From Section III.C from the main paper. 
\begin{proof}

The proof is based on arguments coming from the calculus of variations and is divided into three steps.
We recall the assumptions on $g_t$: $g_t:\mathbb{R}^+ \rightarrow \mathbb{R}^+$, $g_t(0)=1$, $g_t$ is strictly decreasing, $\underset{r\rightarrow +\infty}{\lim}\,g(r)=0$, and there exists $c>0$ such that $c\leq g_t\leq 1$ everywhere. We also have $\phi_t:\bar{\Omega} \rightarrow\bar{\Omega}$ thanks to Ball's results \cite{ball_1981} as seen later. 

\medskip
{\bf{Coercivity inequality}}: We first have that $G(0,(\mbox{Id})_{t=1,\cdots,T})=4a_1\mbox{meas}(\Omega) + \frac{1}{2T}\underset{t=1}{\overset{T}{\sum}}\|x_t\|_{L^2(\mathbb{R}^2)}^2<+\infty$. Let $u \in BV(\Omega')$, $\phi_t\in \mathcal{W},\, \forall t\in \{1,\cdots,T\}$ such that $(\mathcal{C}u)\circ \phi_t \in BV_{g_t,0}(\Omega),\forall t\in \{1,\cdots,T\}$, we then derive a coercivity inequality:
  \begin{align*}
  &G(u,(\phi_t)_{t=1,\cdots,T}) \geq \frac{1}{T}\underset{t=1}{\overset{T}{\sum}} (a_1\|\nabla \phi_t\|_{L^4(\Omega,M_2(\mathbb{R}))}^4 \\
  &+ \frac{a_2}{2}\|\mbox{det}\nabla \phi_t\|_{L^4(\Omega)}^4 + \frac{a_2}{2}\|\frac{1}{\mbox{det}\nabla \phi_t}\|_{L^4(\Omega)}^4\\
  &+\frac{1}{4} \|(\mathcal{C}u)\circ \phi_t^{-1}\|_{L^2(\Omega)}^2+\delta \operatorname{TV}_{g_t}((\mathcal{C}u)\circ \phi_t^{-1})+\alpha \operatorname{TV}(u)\\
  &-\frac{1}{2}\|x_t\|_{L^2(\mathbb{R}^2)}^2-3a_2\mbox{meas}(\Omega)).
  \end{align*}
  Indeed, $((\mathcal{C}u)\circ \phi_t^{-1})\in BV_{g_t,0}(\Omega)\subset BV(\Omega)\subset L^2(\Omega)$ (\cite{debroux-bib:baldi}) with continuous embedding, and $((\mathcal{C}u)\circ \phi_t^{-1})_e$ is the extension by $0$ outside the domain $\Omega$ of $(\mathcal{C}u)\circ\phi_t^{-1}$, then $\|((\mathcal{C}u)\circ \phi_t^{-1})_e\|_{L^2(\mathbb{R}^2)} = \|(\mathcal{C}u)\circ\phi_t^{-1}\|_{L^2(\Omega)}<+\infty$, and finally Plancherel's theorem gives us $\|\mathcal{F}((\mathcal{C}u)\circ \phi_t^{-1})_e\|_{L^2(\mathbb{R}^2)}^2\leq \|\mathcal{A}((\mathcal{C}u)\circ \phi_t^{-1})_e\|_{L^2(\mathbb{R}^2)}^2 = \|((\mathcal{C}u)\circ\phi_t^{-1})_e\|_{L^2(\mathbb{R}^2)}^2=\|(\mathcal{C}u)\circ \phi_t^{-1}\|_{L^2(\Omega)}^2$.\\
  The infimum is thus finite. Let $(u_n,\,(\phi_{t,n})_{t=1,\cdots,T})_n \in \{u\in BV(\Omega'),\, \phi_t\in \mathcal{W},\, \forall t=1,\cdots,T\,|\, (\mathcal{C}u)\circ\phi_t^{-1}\in BV_{g_t,0}(\Omega),\, \forall t\in\{1,\cdots,T\}\}$ be a minimizing sequence such that $\underset{n\rightarrow +\infty}{\lim}G(u_n,(\phi_{t,n})_{t=1,\cdots,N})=\underset{(u,(\phi_t)_{t=1,\cdots,T})\in \mathcal{U}}{\inf} G(u,(\phi_t)_{t=1,\cdots,T})$.
  
\medskip {\bf{Extraction of converging subsequences}}: From the previous coercivity inequality and the finiteness of the infimum we deduce that:
  \begin{itemize}
  \item $(\phi_{t,n})_n$ is uniformly bounded according to $n$ in $W^{1,4}(\Omega,\mathbb{R}^2)$ for each $t=1,\cdots,T$ by using the generalized Poincar\'e's inequality and the fact that $\phi_{t,n}$ is equal to the identity on the boundary $\partial \Omega$. 
  \item $(\mbox{det}\nabla \phi_{t,n})_n$ is uniformly bounded according to $n$ in $L^4(\Omega)$ for each $t=1,\cdots,T$.
  \item $((\mathcal{C}u_n)\circ \phi_{t,n})_n$ is uniformly bounded according to $n$ in $BV_{g_t}(\Omega)$ and thus in $BV(\Omega)$ since $c\leq g_t\leq 1$ everywhere for each $t=1,\cdots,T$.
  \end{itemize}
  Therefore, we can extract subsequences (but for the sake of conciseness we keep the same notations) such that: 
  \begin{itemize}
  \item $\phi_{t,n} \underset{n\rightarrow +\infty}{\rightharpoonup} \bar{\phi}_t$ in $W^{1,4}(\Omega,\mathbb{R}^2)$ for each $t=1,\cdots,T$ and by continuity of the trace operator, we deduce that $\bar{\phi}_t\in \mbox{Id}+W^{1,4}_0(\Omega,\mathbb{R}^2)$.
  \item $\mbox{det}\nabla \phi_{t,n}\underset{n\rightarrow +\infty}{\rightharpoonup} \delta_t$ in $L^4(\Omega)$ for all $t=1,\cdots,T$.
  \item $((\mathcal{C}u_n)\circ \phi_{t,n}^{-1}) \underset{n\rightarrow +\infty}{\longrightarrow} \alpha_t $ in $L^q(\Omega)$ for $q\in [1,2[$ and thus in $L^1(\Omega,g_t)$ with $\alpha_t\in BV(\Omega) \subset BV_{g_t}(\Omega)$. By continuity of the trace operator, we deduce that $\alpha_t\in BV_{g_t,0}(\Omega)$.
  \item $((\mathcal{C}u_n)\circ \phi_{t,n}^{-1}) \underset{n\rightarrow +\infty}{\rightharpoonup} \alpha_t$ in $L^2(\Omega)$ for each $t=1,\cdots,T$ by uniqueness of the weak limit in $L^1(\Omega)$ and the continuous embedding of $L^2(\Omega) \subset L^1(\Omega)$. Also, since $\alpha_t\in BV_{g_t,0}(\Omega)$, we can extend it by $0$ outside the domain $\Omega$ and we have $((\mathcal{C}u_n)\circ \phi_{t,n}^{-1})_e \underset{n\rightarrow +\infty}{\rightharpoonup} (\alpha_t)_e$ in $L^2(\mathbb{R}^2)$, for each $t=1,\cdots,T$.
  \end{itemize}
  Also, from \cite[Theorem 8.20]{debroux-bib:dacorogna}, we have that $\mbox{det}\nabla \phi_{t,n} \underset{n\rightarrow +\infty}{\rightharpoonup} \mbox{det}\nabla \bar{\phi}_t$ in $L^2(\Omega)$ for each $t=1,\cdots,T$, and by uniqueness of the weak limit in $L^2(\Omega)$ and the continuous embedding $L^4(\Omega)\subset L^2(\Omega)$, we deduce that $\delta_t=\mbox{det}\nabla \bar{\phi}_t$ for each $t=1,\cdots,T$. Also, since we have for $\tilde q=2+\frac{1}{2}>2$ : 
  \begin{align*}
  & \int_\Omega \|(\nabla \phi_{t,n})^{-1}(x)\|_F^{\tilde q}\mbox{det}\nabla \phi_{t,n}(x)\,dx\\
  &= \int_\Omega \|\nabla \phi_{t,n}(x)\|_F^{\tilde q}(\mbox{det}\nabla \phi_{t,n})^{1-\tilde q}\,dx,\\
   &\leq \|\nabla \phi_{t,n}\|_{L^4(\Omega,M_2(\mathbb{R}))}^{\tilde q} \left\|\frac{1}{\mbox{det}\nabla \phi_{t,n}}\right\|_{L^4(\Omega)}^{\tilde q-1},
  \end{align*}
using H\"older's inequality with $\tilde p=\frac{4}{\tilde q}$, $r=\frac{4}{\tilde q-1}$ and noticing that $\frac{4(\tilde q-1)}{4-\tilde q} = 4$. This quantity is uniformly bounded according to $n$ from what precedes and we deduce from \cite[Theorem 1 and 2]{corona-bib:ball} that $(\phi_{t,n})$ are bi-H\"older's homeomorphisms and therefore $\phi_{t,n}^{-1}$ exists.

\medskip
{\bf{Weak lower semi-continuity}}: $W_{Op}$ is convex and continuous. If $\psi_n\underset{n\rightarrow +\infty}{\longrightarrow}\bar{\psi}$ in $W^{1,4}(\Omega,\mathbb{R}^2)$, then $\nabla \psi_n \underset{n\rightarrow +\infty}{\longrightarrow} \nabla \bar{\psi}$ in $L^4(\Omega,M_2(\mathbb{R}))$ and we can extract a subsequence still denoted $(\nabla \psi_n)$ such that $\nabla \psi_n \underset{n\rightarrow +\infty}{\longrightarrow} \nabla \bar{\psi}$ almost everywhere in $\Omega$. If $\delta_n \underset{n\rightarrow +\infty}{\longrightarrow} \bar{\delta}$ in $L^4(\Omega)$, then we can extract a subsequence still denoted $(\delta_n)$ such that $\delta_n \underset{n\rightarrow +\infty}{\longrightarrow}\bar{\delta}$ almost everywhere in $\Omega$. Then by continuity of $W_{Op}$, we get that $W_{Op}(\nabla \psi_n(x),\delta_n(x)) \underset{n\rightarrow +\infty}{\longrightarrow} W_{Op}(\nabla \bar{\psi}(x),\bar{\delta}(x))$ almost everywhere in $\Omega$. Then by applying Fatou's lemma, we have that $\underset{n\rightarrow +\infty}{\lim \inf} \int_\Omega W_{Op}(\nabla \psi_n(x),\delta_n(x))\,dx \geq \int_\Omega W_{Op}(\nabla \bar{\psi}(x),\bar{\delta}(x))\,dx$. Since $W_{Op}$ is convex, so is $\int_\Omega W_{Op}(\xi(x),\delta(x))\,dx$, and we can apply \cite[Corollaire III.8]{debroux-bib:brezis} to get that $\int_\Omega W_{Op}(\xi(x),\delta(x))\,dx$ is lower semicontinuous in $L^4(\Omega,M_2(\mathbb{R}))\times L^4(\Omega)$. We deduce that $\underset{n\rightarrow +\infty}{\lim \inf}\, \int_\Omega W_{Op}(\nabla \phi_{t,n}(x), \mbox{det}\nabla \phi_{t,n}(x))\,dx\geq \int_\Omega W_{Op}(\nabla \bar{\phi}_t(x),\mbox{det}\nabla \bar{\phi}_t(x))\,dx$.\\
 Also since $(\mathcal{C}u_n)\circ \phi_{t,n}^{-1} \underset{n\rightarrow +\infty}{\longrightarrow} \alpha_t$ in $L^1(\Omega)$ and so in $L^1(\Omega,g_t)$, then by the semi-continuity theorem from \cite[Theorem 3.2]{debroux-bib:baldi}, we conclude that $\operatorname{TV}_{g_t}(\alpha_t)\leq \underset{n\rightarrow +\infty}{\lim \inf}\, \operatorname{TV}_{g_t}((\mathcal{C}u_n)\circ \phi_{t,n}^{-1})$ for all $t=1,\cdots,T$.\\
 $\mathcal{F}$ is a linear operator and continuous for the strong topology from $L^2(\mathbb{R}^2)$ to $L^2(\mathbb{R}^2)$. Therefore, by applying \cite[Theorem III.9]{debroux-bib:brezis}, $\mathcal{F}$ is continuous from the weak topology of $L^2(\mathbb{R}^2)$ to the weak topology of $L^2(\mathbb{R}^2)$. As $((\mathcal{C}u_n)\circ \phi_{t,n})_e \underset{n\rightarrow +\infty}{\rightharpoonup} (\alpha_t)_e$ in $L^2(\mathbb{R}^2)$, we deduce that $\mathcal{F}((\mathcal{C}u_n)\circ \phi_{t,n})_e\underset{n\rightarrow +\infty}{\rightharpoonup} \mathcal{F}(\alpha_t)_e$ and thus $\mathcal{F} ((\mathcal{C}u_n) \circ \phi_{t,n})_e - x_t \underset{n\rightarrow +\infty}{\rightharpoonup} \mathcal{F}(\alpha_t)_e-x_t $ in $L^2(\mathbb{R}^2)$. By the lower semi-continuity of the norm, we deduce that $\underset{n\rightarrow +\infty}{\lim \inf} \|\mathcal{F}((\mathcal{C}u_n)\circ \phi_{t,n})_e-x_t\|_{L^2(\mathbb{R}^2)}^2 \geq \|\mathcal{F}(\alpha_t)_e-x_t\|_{L^2(\mathbb{R}^2)}^2$.\\
  We now need to prove that $(\mathcal{C}u_n)\circ \phi_{t,n}^{-1}\circ \phi_{t,n}=\mathcal{C}u_n \underset{n\rightarrow +\infty}{\longrightarrow}\alpha_t\circ \bar{\phi}_t=\bar{U}$ in $L^p(\Omega)$ for all $t=1,\cdots,T$. We first have :
 \begin{align*}
  &\|(\mathcal{C}u_n)\circ \phi_{t,n}^{-1}\circ \phi_{t,n} - \alpha_t \circ \bar{\phi}_t\|_{L^p(\Omega)} \\
  &\leq \|(\mathcal{C}u_n) \circ \phi_{t,n}^{-1} \circ \phi_{t,n} - \alpha_t \circ \phi_{t,n}\|_{L^p(\Omega)} \\
  &+ \|\alpha_t \circ \phi_{t,n} - \alpha_t \circ \bar{\phi}_t\|_{L^p(\Omega)}.
  \end{align*}
  We now focus on the first term and make the change of variable $y= \phi_{t,n}(x) \Leftrightarrow x= \phi_{t,n}^{-1}(y)$ and $dy = \mbox{det}\nabla \phi_{t,n}(x)\,dx,\, dx = \frac{1}{\mbox{det}\nabla \phi_{t,n}(\phi_{t,n}^{-1}(y))}\,dy$ : 
  \begin{align*}
  &\int_\Omega |(\mathcal{C}u_n) \circ \phi_{t,n}^{-1} \circ \phi_{t,n} - \alpha_t \circ \phi_{t,n}|^p\,dx \\
  &= \int_\Omega |(\mathcal{C}u_n) \circ \phi_{t,n}^{-1} - \alpha_t |^p\frac{1}{|\mbox{det}\nabla \phi_{t,n}(\phi_{t,n}^{-1}(y))|}\,dy,\\
  &\leq \|(\mathcal{C}u_n)\circ \phi_{t,n}^{-1}-\alpha_t\|_{L^{\frac{5p}{4}}(\Omega)}^p(\int_\Omega \frac{1}{|\mbox{det}\nabla \phi_{t,n}(\phi_{t,n}^{-1}(y))|^5}\,dy)^{\frac{1}{5}},\\
  &\leq \|(\mathcal{C}u_n)\circ \phi_{t,n}^{-1}-\alpha_t\|_{L^{\frac{5p}{4}}(\Omega)}^p(\int_\Omega \frac{1}{|\mbox{det}\nabla \phi_{t,n}(x))|^4}\,dx)^{\frac{1}{5}},
  \end{align*}
  using H\"older's inequality and making another change of variable. Since $ (\mathcal{C}u_n)\circ \phi_{t,n}^{-1}\underset{n\rightarrow +\infty}{\longrightarrow}\alpha_t$ in $L^{\frac{5p}{4}}(\Omega)$, as $\frac{5p}{4}<2$, and $\|\frac{1}{\mbox{det}\nabla \phi_{t,n}}\|_{L^4(\Omega)}$ is uniformly bounded, we deduce that $ \int_\Omega |(\mathcal{C}u_n )\circ \phi_{t,n}^{-1} \circ \phi_{t,n} - \alpha_t \circ \phi_{t,n}|\,dx \underset{n\rightarrow +\infty}{\longrightarrow}0$.\\
  According to \cite[Theorem 6.70]{corona-bib:demengel}, there exists a sequence $(\xi_t^k)_k$ of functions in $\mathcal{C}^\infty_c(\Omega)$ such that $\|\alpha_t - \xi_t^k\|_{L^1(\Omega)} \underset{k\rightarrow +\infty}{\longrightarrow} 0$ and $\int_\Omega |\nabla \xi_t^k| \underset{k\rightarrow +\infty}{\longrightarrow} \int_\Omega |\nabla \alpha_t| + \int_{\partial \Omega } |\alpha_t|\,dx$ with $\int_{\partial \Omega } |\alpha_t|\,dx=0$ since $\alpha_t=0$ on $\partial \Omega$, for all $t=1,\cdots,T$. Thus $(\xi_t^k)$ is uniformly bounded according to $k$ in $BV(\Omega)\underset{c}{\subset} L^q(\Omega) $, for $q\in [1,2[$, and therefore $\xi_t^k \underset{k\rightarrow +\infty}{\longrightarrow} \alpha_t$ in $L^q(\Omega)$ for $q\in [1,2[$. Let $\varepsilon>0$. Thus we fix $N\in \mathbb{N}^*$ such that $\|\xi_t^N-\alpha_t\|_{L^{\frac{5p}{4}}(\Omega)}\leq \frac{\varepsilon}{3}$. We now have :
  \begin{align*}
  &\|\alpha_t \circ \phi_{t,n} - \alpha_t \circ \bar{\phi}_t\|_{L^p(\Omega)}\leq \|\alpha_t \circ \phi_{t,n} - \xi_t^N \circ \phi_{t,n}\|_{L^p(\Omega)} \\
  &+ \|\xi_t^N \circ \phi_{t,n} - \xi_t^N\circ \bar{\phi}_t\|_{L^p(\Omega)}+ \|\xi_t^N \circ \bar{\phi}_{t} - \alpha_t \circ \bar{\phi}_t\|_{L^p(\Omega)},\\
  &\leq (\int_\Omega |\alpha_t -\xi_t^N|^p\frac{1}{|\mbox{det}\nabla \phi_{t,n}(\phi_{t,n}^{-1}(y))|}\,dy)^{\frac{1}{p}} \\
  &+ L_\varepsilon\|\phi_{t,n} - \bar{\phi}_t\|_{L^p(\Omega,\mathbb{R}^2)} \\
  &+ (\int_\Omega |\xi_t^N-\alpha_t|^p\frac{1}{|\mbox{det}\nabla \bar{\phi}_t(\bar{\phi}_t^{-1}(y))|}\,dy)^{\frac{1}{p}},
 \end{align*}
 with $L_\varepsilon$ the Lipschitz constant of $\xi_t^N$ since it belongs to $\mathcal{C}^\infty_c(\Omega)$ and so is Lipschitz continuous, and using a change of variable. As $\phi_{t,n} \underset{n \rightarrow +\infty}{\rightharpoonup} \bar{\phi}_t$ in $W^{1,4}(\Omega) \underset{c}{\subset} L^p(\Omega,\mathbb{R}^2)$, there exists $K\in \mathbb{N}^*$ such that for any $n\geq K$, $\|\phi_{t,n} -\bar{\phi}_t\|_{L^p(\Omega,\mathbb{R}^2)}\leq \frac{\varepsilon}{3L_\varepsilon}$. From now on, we assume $n$ satisfies $n\geq K$, and we use H\"older's inequality : 
 \begin{align*}
   &\|\alpha_t \circ \phi_{t,n} - \alpha_t \circ \bar{\phi}_t\|_{L^p(\Omega)}\leq \|\alpha_t -\xi_t^N\|_{L^{\frac{5p}{4}}(\Omega)}\\
   &(\int_\Omega\frac{1}{|\mbox{det}\nabla \phi_{t,n}(\phi_{t,n}^{-1}(y))|^{5}}\,dy)^{\frac{1}{5p}} + \frac{\varepsilon}{3} \\
   &+\|\xi_t^N-\alpha_t\|_{L^{\frac{5p}{4}}(\Omega)}(\int_\Omega \frac{1}{|\mbox{det}\nabla \bar{\phi}_t(\bar{\phi}_t^{-1}(y))|^5}\,dy)^{\frac{1}{5p}},\\
    &\leq \frac{\varepsilon}{3}\|\frac{1}{\mbox{det}\nabla \phi_{t,n}}\|_{L^4(\Omega)}^{\frac{4}{5p}} + \frac{\varepsilon}{3} +\frac{\varepsilon}{3}\|\frac{1}{\mbox{det}\nabla \bar{\phi}_t}\|_{L^4(\Omega)}^{\frac{4}{5p}},
   \end{align*}
   with $\|\frac{1}{\mbox{det}\nabla \phi_{t,n}}\|_{L^4(\Omega)}$ uniformly bounded according to $n$ and $\| \frac{1}{\mbox{det}\nabla \bar{\phi}_t}\|_{L^4(\Omega)}$ bounded from what precedes. So by letting $\varepsilon$ tend to $0$, we obtain that $ \int_\Omega |\alpha_t \circ \phi_{t,n} - \alpha_t \circ \bar{\phi}_t|^p\,dx \underset{n\rightarrow +\infty}{\longrightarrow}0$ and consequently $ (\mathcal{C}u_n)\circ \phi_{t,n}^{-1} \circ \phi_{t,n} = \mathcal{C}u_n \underset{n\rightarrow+\infty}{\longrightarrow} \alpha_t \circ \bar{\phi}_t$ for all $t=1,\cdots,N$. By uniqueness of the limit, we have that $\bar{U}=\alpha_t\circ\bar{\phi}_t \Leftrightarrow \alpha_t=\bar{U}\circ \bar{\phi}_t^{-1}$ for all $t=1,\cdots,T$ in $L^p(\Omega)$, with $\bar{U}\in L^p(\Omega)$ and $ \bar{U}\circ \bar{\phi}_t^{-1} \in BV_{g_t,0}(\Omega)$. \\
 We now set $\tilde u_n = \frac{1}{|\Omega'|}\int_{\Omega'}u_n\,dx$, and $u_{0,n}=u_n-\tilde u_n$. We clearly have $\int_{\Omega'}u_{0,n}\,dx=0$ for all $n$, and $\operatorname{TV}(u_{0,n})=\operatorname{TV}(u_n)$ is uniformly bounded thanks to the coercivity inequality. We denote this uniform bound by $\nu$. Thus using Poincar\'e-Wirtinger's inequality we obtain
 \begin{align*}
  \|u_{0,n}\|_{L^1(\Omega')}\leq c_1 \operatorname{TV}(u_{0,n})\leq c_1\nu,
 \end{align*}
 with $c_1>0$. We now need a bound for $\|\mathcal{C}\tilde u_n\|_{L^p(\Omega)}$:
 \begin{align*}
 & \|\mathcal{C}\tilde u_n\|_{L^p(\Omega)}^2-2\|\mathcal{C}\tilde u_n\|_{L^p(\Omega)}\|\mathcal{C}\|_{\mathcal{L}(L^1(\Omega'),L^p(\Omega))}\|u_{0,n}\|_{L^1(\Omega')}\\
 &\leq \|\mathcal{C}\tilde u_n\|_{L^p(\Omega)}^2-2\|\mathcal{C}\tilde u_n\|_{L^p(\Omega)}\|\mathcal{C}u_{0,n}\|_{L^p(\Omega)},\\
  &\leq (\|\mathcal{C}\tilde u_n\|_{L^p(\Omega)}-\|\mathcal{C}u_{0,n}\|_{L^p(\Omega)})^2,\\
  &\leq \|\mathcal{C}(\tilde u_n + u_{0,n})\|_{L^p(\Omega)}^2,\\
  &\leq \|\mathcal{C}u_n\|_{L^p(\Omega)}^2 \leq c_2^2<+\infty,
 \end{align*}
 as $\mathcal{C}u_n$ strongly converges in $L^p(\Omega)$. Besides, we have $ \|\mathcal{C}\|_{\mathcal{L}(L^1(\Omega'),L^p(\Omega))}\|u_{0,n}\|_{L^1(\Omega')} \leq \|\mathcal{C}\|_{\mathcal{L}(L^1(\Omega'),L^p(\Omega))}c_1\nu=c_3<+\infty$ and thus $\|\mathcal{C}\tilde u_n\|_{L^p(\Omega)} \leq 2\sqrt{c_3^2+c_2^2}$. But we also know that $\|\mathcal{C}\tilde u_{n}\|_{L^p(\Omega)} = \frac{1}{|\Omega'|}|\int_{\Omega'} u_n\,dx|\|\mathcal{C}\mathbb{1}\|_{L^p(\Omega)}\leq 2\sqrt{c_2^2+c_3^2}$. Since $\mathcal{C}\mathbb{1}\neq 0$, we have $\frac{1}{|\Omega'|}|\int_{\Omega'}u_n\,dx|\leq \frac{2\sqrt{c_2^2+c_3^2}}{\|\mathcal{C}\mathbb{1}\|_{L^p(\Omega)}}=c_4<+\infty $. We therefore have
 \begin{align*}
  \|u_n\|_{L^1(\Omega')}&\leq \|u_{0,n}+\frac{1}{|\Omega'|}|\int_{\Omega'}u_n\,dx|\|_{L^1(\Omega')},\\
  &\leq \|u_{0,n}\|_{L^1(\Omega')}+|\int_{\Omega'}u_n\,dx|,\\
  &\leq c_1\nu+c_4|\Omega'|<+\infty.
 \end{align*}
 Thus $u_n$ is uniformly bounded according to $n$ in $BV(\Omega')$ and there exists a subsequence still denoted $(u_n)$ such that $u_n\underset{n\rightarrow +\infty}{\longrightarrow}\bar{u}$ in $L^1(\Omega')$ with $\bar{u}\in BV(\Omega')$. By continuity of the operator $\mathcal{C}$ and the uniqueness of the limit, we deduce that $\mathcal{C}u_n \underset{n\rightarrow +\infty}{\longrightarrow} \mathcal{C}\bar{u} = \bar{U}$ in $L^p(\Omega)$. By the semi-continuity theorem, we get $\operatorname{TV}(\bar{u}) \leq \underset{n\rightarrow +\infty}{\lim\inf}\, \operatorname{TV}(u_n) $.\\ 
    By combining all the results, we obtain that $+\infty > \underset{(u,(\phi_t)_{t=1,\cdots,T})\in \mathcal{U}}{\inf} G(u,(\phi_t)_{t=1,\cdots,T}) = \underset{n\rightarrow +\infty}{\lim \inf}\, G(u_n,(\phi_{t,n})_{t=1,\cdots,T}) \geq \frac{1}{T}\underset{t=1}{\overset{T}{\sum}}\delta \operatorname{TV}_{g_t}((\mathcal{C}\bar{u})\circ\bar{\phi}_t^{-1}) + \frac{1}{2}\|\mathcal{F}((\mathcal{C}\bar{u})\circ \bar{\phi}_t^{-1})_e-x_t\|_{L^2(\mathbb{R}^2)}^2+\int_\Omega W_{Op}(\nabla \bar{\phi}_t)\,dx + \alpha \operatorname{TV}(\bar{u})$.
  It thus means that $\mbox{det}\nabla \bar{\phi}_t\in L^4(\Omega)$, $\frac{1}{\mbox{det}\nabla \bar{\phi}_t}\in L^4(\Omega)$, and $\mbox{det}\nabla \bar{\phi}_t >0$ almost everywhere in $\Omega$ for all $t=1,\cdots,T$. Indeed, since $W_{Op}(\nabla \bar{\phi}_t(x),\mbox{det}\nabla \bar{\phi}_t(x))=+\infty$ when $\mbox{det}\nabla \bar{\phi}_t(x)\leq 0$, it means that the set on which it happens must be of null measure otherwise we would have $ \int_\Omega W_{Op}(\nabla \bar{\phi}_t,\mbox{det}\nabla \bar{\phi}_t)\,dx = +\infty$. Also, by applying the same reasoning for each $\phi_{t,n}$, we prove that $\bar{\phi}_t$ is a bi-H\"older homeomorphism and have that $\bar{\phi}_t\in \mathcal{W}$ for each $t=1,\cdots,T$.\\ 
  We thus have proved the existence of minimisers for our problem (4).

\end{proof}

\end{document}